\documentclass[11pt,letterpaper]{article}
\usepackage[margin=1.1in]{geometry}

\usepackage{amsmath,amsfonts,amssymb,bbm,amsthm} 
\usepackage{natbib}
\newlength{\dhatheight}
\usepackage{thm-restate}
\usepackage[symbol]{footmisc}
\usepackage{enumerate}
\usepackage{hyperref}
\usepackage{cleveref}
\usepackage{xcolor}

\Crefname{conditions}{Cond.}{Conds.}
\Crefformat{conditions}{#2C#1#3}
\crefformat{conditions}{#2C#1#3}

\newcommand{\ex}[1]{\mathbb E\left[#1\right]}

\newenvironment{numberedtheorem}[1]{%
	\renewcommand{\thetheorem}{#1}%
	\begin{theorem}}{\end{theorem}\addtocounter{theorem}{-1}}

\newenvironment{numberedlemma}[1]{%
	\renewcommand{\thelemma}{#1}%
	\begin{lemma}}{\end{lemma}\addtocounter{lemma}{-1}}

\newtheorem{claim}{Claim}
\newtheorem{lemma}[claim]{Lemma}
\newtheorem{theorem}[claim]{Theorem}

\newtheorem{definition}[claim]{Definition}
\newtheorem{corollary}[claim]{Corollary}

\usepackage{ifthen}
\DeclareMathOperator{\argmax}{argmax}

\begin{document}
	
	\begin{titlepage}
		\title{Sample Complexity for Non-truthful Mechanisms}
		
		\author{Jason Hartline\footnote{Northwestern University, hartline@eecs.northwestern.edu}\,\, and Samuel Taggart\footnote{Oberlin College, Sam.Taggart@oberlin.edu}}
		
		\maketitle
		\begin{abstract}

This paper considers the design of non-truthful mechanisms from
samples.  We identify a parameterized family of mechanisms with
strategically simple winner-pays-bid, all-pay, and truthful payment
formats.  In general (not necessarily downward-closed)
single-parameter feasibility environments we prove that the family has
low representation and generalization error.  Specifically,
polynomially many bid samples suffice to identify and run a mechanism
that is $\epsilon$-close in Bayes-Nash equilibrium revenue or welfare
to that of the optimal truthful mechanism with high probability.


		\end{abstract}

	\end{titlepage}

	\renewcommand*{\thefootnote}{\arabic{footnote}}

\newcommand{\setsize}[1]{{\left|#1\right|}}

\newcommand{\floor}[1]{
{\lfloor {#1} \rfloor}
}
\newcommand{\bigfloor}[1]{
{\left\lfloor {#1} \right\rfloor}
}

\newcommand{\super}[1]{^{(#1)}}

\newcommand{\inteval}[1]{\Big[#1\Big]}
\newcommand{\tinteval}[1]{\left[#1\right]}
\newcommand{\distcond}[2]{\left.#1\vphantom{\big|}\right|_{#2}}
\newcommand{\evalat}[2]{\left.#1\vphantom{\big|}\right|_{#2}}

%
%
\newcommand{\given}{\,\mid\,}

\newcommand{\prob}[2][]{\text{\bf Pr}\ifthenelse{\not\equal{}{#1}}{_{#1}}{}\!\left[{\def\givenn{\middle|}#2}\right]}
\newcommand{\expect}[2][]{\text{\bf E}\ifthenelse{\not\equal{}{#1}}{_{#1}}{}\!\left[{\def\givenn{\middle|}#2}\right]}

\newcommand{\tparen}{\big}
\newcommand{\tprob}[2][]{\text{\bf Pr}\ifthenelse{\not\equal{}{#1}}{_{#1}}{}\tparen[{\def\given{\tparen|}#2}\tparen]}
\newcommand{\texpect}[2][]{\text{\bf E}\ifthenelse{\not\equal{}{#1}}{_{#1}}{}\tparen[{\def\given{\tparen|}#2}\tparen]}

\newcommand{\sprob}[2][]{\text{\bf Pr}\ifthenelse{\not\equal{}{#1}}{_{#1}}{}[#2]}
\newcommand{\sexpect}[2][]{\text{\bf E}\ifthenelse{\not\equal{}{#1}}{_{#1}}{}[#2]}

\newcommand{\suchthat}{\,:\,}

\newcommand{\partialx}[2][]{{\tfrac{\partial #1}{\partial #2}}}
\newcommand{\nicepartialx}[2][]{{\nicefrac{\partial #1}{\partial #2}}}
\newcommand{\dd}{{\mathrm d}}
\newcommand{\ddx}[2][]{{\tfrac{\dd #1}{\dd #2}}}
\newcommand{\niceddx}[2][]{{\nicefrac{\dd #1}{\dd #2}}}
\newcommand{\grad}{\nabla}

\newcommand{\symdiff}{\triangle}

	\newcommand{\run}{\text{run}}
\newcommand{\design}{\text{design}}
\newcommand{\polyd}{p_{\design}}
\newcommand{\polyr}{p_{\run}}
\newcommand{\numd}{m_{\design}}
\newcommand{\numr}{m_{\run}}

\newcommand{\num}{T}

\newcommand{\numag}{n}

\newcommand{\agind}[1][i]{_{#1}}
\newcommand{\noaccents}[1]{#1}
\newcommand{\marginal}[1]{#1'}

\newcommand{\newagentvar}[3][\noaccents]{%
\expandafter\newcommand\expandafter{\csname #2\endcsname}{#1{#3}}%
\expandafter\newcommand\expandafter{\csname #2s\endcsname}{#1{\boldsymbol{#3}}}%
\expandafter\newcommand\expandafter{\csname #2smi\endcsname}[1][i]{#1{\boldsymbol{#3}}_{-##1}}%
\expandafter\newcommand\expandafter{\csname #2i\endcsname}[1][i]{#1{#3}\agind[##1]}%
\expandafter\newcommand\expandafter{\csname #2ith\endcsname}[1][i]{#1{#3}_{(##1)}}%
}

\newagentvar{quant}{q}
\newagentvar{val}{v}
\newagentvar{price}{p}
\newagentvar{bid}{b}
\newagentvar{alloc}{x}
\newagentvar{util}{u}
\newagentvar[\tilde]{balloc}{\alloc}
\newagentvar[\tilde]{bprice}{\price}
\newagentvar[\hat]{epalloc}{\alloc}
\newagentvar[\bar]{ralloc}{\alloc}
\newagentvar{dist}{F}
\newagentvar{dens}{f}
\newagentvar{virt}{\phi}
\newagentvar{cumvirt}{\Phi}
\newagentvar{bdist}{G}
\newagentvar{Sur}{\Psi}
\newagentvar{sur}{\psi}
\newagentvar{Wal}{W}
\newagentvar{wal}{w}
\newagentvar{rank}{r}
\newagentvar{strat}{s}
\newagentvar{rev}{R}
\newagentvar[\marginal]{marg}{\rev}
\newagentvar[\hat]{esur}{\sur}
\newagentvar[\hat]{eSur}{\Sur}
\newagentvar[\bar]{isur}{\sur}
\newagentvar[\bar]{iSur}{\Sur}
\newagentvar[\bar]{ialloc}{\alloc}
\newagentvar[\hat]{ealloc}{\alloc}

\newagentvar{murev}{P}
\newagentvar[\bar]{imurev}{\murev}

\newagentvar{error}{\epsilon}

	\newcommand{\st}[1]{\color{red} #1 \color{black}}

\section{Introduction}
%
%
The classical theory of revenue-maximizing mechanism design requires
knowledge of agents' value distributions.  As a result, the sample
complexity of revenue maximization has received significant attention
in recent years.  This work has placed Bayesian mechanism design on
more practical footing by analyzing the amount of sampled data
necessary to produce a nearly-optimal mechanism.  A key assumption is
that samples are agents' values;  e.g.\ they are obtained from past
runs of a truthful mechanism.  Moreover, the mechanisms produced are
themselves truthful, and hence generate data suitable for future
inference.

For many applications, however, truthful data may not be available.
Samples are often equilibrium bids, produced by mechanisms where
agents are not incentivized to report their values.  Moreover,
practical constraints often require a designer to employ non-truthful
payment formats like those of the first-price (i.e.\ winner-pays-bid)
or all-pay auction.  This paper develops a theory of non-truthful
mechanism design from samples.  We identify a family of non-truthful
mechanisms that have near-optimal revenue or welfare and require only
polynomially many samples to design and implement.  Our mechanisms'
performance guarantees hold in equilibrium under standard non-truthful
(or truthful) payment formats.  We assume sampled data come from
mechanisms within our family, and therefore need not be truthful.  Our
mechanisms may thus be redesigned as necessary when the environment
changes.

%
%

%
%
Both practical and theoretical considerations necessitate the study of
non-truthful mechanism design. In some common applications, outcomes
are contracts, e.g., government procurement auctions, variable
commission mechanisms of third-party listing agencies like
Booking.com, and ad exchanges. For these applications, the theory of
winner-pays-bid mechanisms is most appropriate.  For games of effort
-- like crowdsourcing contests \citep[e.g.,][]{CHS-15}, forecasting
\citep[e.g.,][]{osb-89}, and peer prediction \citep[e.g.,][]{DG-13} --
the theory of all-pay mechanisms is most appropriate.  \citet{AM06}
discuss a number of other pragmatic concerns and describe why truthful
mechanisms are rarely seen in practice.  These observations justify
revisiting standard theoretical questions in mechanism design (such as
sample complexity) with non-truthfulness imposed as an exogenous
constraint.

Non-truthful mechanisms can also exhibit surprising theoretical advantages.
For
example, \citet{FH-18} proved that non-truthful mechanisms can be
strictly better than truthful mechanisms when robustness to
distributional assumptions is desired.  In other words, the so-called
``revelation principle'' fails when looking for simple and robust
mechanisms. While understanding the theoretical limits of non-truthful mechanisms is a potentially fruitful avenue of future research, such studies are hindered by the sparse literature on designing such mechanisms.
Non-truthful mechanisms add two significant complications to the study
of sample complexity.  First, for non-truthful mechanisms it should be
assumed that there is sample access to equilibrium bids rather than
values of the agents.  Second, and more fundamentally, the literature
has not yet shown how to design non-truthful mechanisms that have
near-optimal equilibria with only estimates of the value
distribution. (This task is straightforward with full distributional
knowledge; see \Cref{app:unrevelation}.) To circumvent these issues,
we give a black-box reduction from non-truthful mechanism design in
arbitrary single-parameter settings to the design of rank-based
position auctions for i.i.d. agents. In the latter setting, bid
distributions are well-suited for estimation and it suffices to infer
a limited set of parameters to design a near-optimal mechanism. The
resulting approach is compatible with winner-pays-bid, all-pay, or
truthful payment formats.

A theory of sample complexity for non-truthful mechanisms further
requires careful delineation of two aspects of sample complexity:
where samples come from and how samples are used.  We assume samples
are from the bid distributions of non-truthful mechanisms.  A
non-issue for truthful sample complexity, it is important for
non-truthful formats which mechanism the bids are from; the Bayes-Nash
equilibrium bid distributions of distinct non-truthful mechanisms are
generally distinct. Our mechanisms use samples in two ways. First,
samples can be used at design time to select the mechanism to be run,
in which case they are usually data from past mechanisms, and can be
assumed to be observed by the agents before bids are placed. Second,
samples can be used at run time when the mechanism executes, in which
case they are viewed as a source of randomness as agents are
bidding. Run-time samples are most plausible when they are from
concurrent runs of the same mechanism, e.g., in high-frequency
settings such as advertising auctions.  In analyses of mechanisms from
samples, incentives and performance guarantees hold with high
probability in design-time samples and in expectation in run-time
samples.

Both uses have analogs in the literature on truthful mechanism design
from samples. Most papers on sample complexity use design-time samples
and elide computational considerations around producing desirable
allocations and truthful payments. The literature on black-box
reductions from Bayesian incentive compatible mechanism design to
Bayesian algorithm design, meanwhile, relies heavily on run-time
samples to achieve exact incentive compatibility. Even for the
simplest setting of single-parameter agents \citep[]{HL-15} designing
an exactly incentive compatible black-box reduction from design-time
samples alone remains an open problem. A more detailed discussion of
the literatures on sample complexity and black-box reductions, as well
as the way they use samples, can be found in the related work section
below.

\paragraph{Problem Statement} We consider the problem of designing good mechanisms from samples in
general single-parameter environments with independently
distributed values in Bayes-Nash equilibrium where a general set
system governs the subsets of agents that can be simultaneously served.
Truthful mechanisms for this environment are well understood.  The
Vickrey-Clarke-Groves mechanism maximizes welfare, and a
straightforward generalization of \citet{M81} gives the truthful
mechanism that maximizes expected revenue \citep[e.g.,][]{har-13}.  The
sample complexity of truthful mechanisms of this setting was largely resolved
by \citet{DHP16}, \citet{GN-17}, and \citet{GHZ19}.

We generalize truthful sample complexity to non-truthful mechanisms in the following way. The problem of {\em non-truthful
  sample complexity} is to identify in a parameterized family of winner-pays-bid (or all-pay) mechanisms and polynomials $\polyd$ and $\polyr$ such that with
$n$-agent environments and desired loss $\epsilon$:
\begin{enumerate}[C1.]
    \item \label[conditions]{cond:inference} With $\numd =
      \polyd(n,\epsilon^{-1},\delta^{-1})$ design-time samples of profiles of
      Bayes-Nash equilibrium bids from any mechanism in the family,
      parameters of the designed mechanism can be selected.

    \item \label[conditions]{cond:incentives} With $\numr =
      \polyd(n,\epsilon^{-1})$ run-time samples of profiles of Bayes-Nash
      equilibrium bids in the selected mechanism, the selected
      mechanism can be run.
    
    \item\label[conditions]{cond:performance} With probability in the
      $\numd$ design-time samples of at least $1-\delta$, the
      expected performance, in agents' values and the $\numr$ run-time
      samples of the selected mechanism, is at most $\epsilon$ less
      than that of the Bayesian optimal mechanism.\footnote{Multiplicative
        versions of \Cref{cond:performance} are also interesting. In
        the multiplicative version a mechanism with expected welfare or revenue at least $(1-\epsilon)$ times optimal is
        required.}
\end{enumerate}

The following story fits the above problem and is implicit in previous papers on mechanism design from samples. Per a standard interpretation of the Bayesian model for auctions, a designer aims to
run a mechanism on agents drawn from one or several large populations.
A mechanism is sought that performs well on a fresh draw of agents
from each population. Our non-truthful mechanism designer
fixes a large parameterized family of mechanisms and has
independently drawn profiles of historical bids in one mechanism in
the family.  The designer uses these historical bid profiles as
design-time samples to select new parameters of the mechanism.  The
agents adapt to the new equilibrium in the new mechanism.  The
designer collects historical samples in the new mechanism and uses
them as run-time samples in its execution.

As is common in the literature on Bayes-Nash mechanism design, we
consider runtime samples and agent strategies which follow a steady
state equilibrium. While it is beyond the scope of this paper to model
the adaptive process by which the agents might learn this equilibrium,
we shall see that the mechanisms we produce have straightforward
bidding problems for each agent. Similarly, we consider it beyond the
scope of the paper to explicitly model the process by which runtime
samples are obtained. For motivation, however, we note here several
scenarios which justify their use. In practical applications such as
ad auctions, agents bid in advance of the auction, and it is possible
to batch the bid collection for many individual executions of the
mechanism together.  For these batched executions, the run-time
samples can be from the other bid profiles that are collected within
the same batch. If batching is infeasible, it may still be possible to
produce run-time samples in an online manner by taking bid data from the
most recent iterations of the mechanism.

\paragraph{Approach and Results.}

We solve the stated problem for general
single-parameter environments and independent but non-identically
distributed agents.  With $n$ agents, polynomial in $n$,
$\epsilon^{-1}$, and $\delta^{-1}$ design-time samples are sufficient to identify a
mechanism that, with polynomially many run-time samples and probability at least $1-\delta$, approximates the performance of the optimal
mechanism to within precision $\epsilon$ in the following
environments:
\begin{itemize}
	\item (non-truthful) winner-pays-bid and all-pay mechanisms,
          additive welfare approximation, and bounded value
          distributions;\footnote{Because of asymmetries in the value distributions of agents and the feasibility environment, welfare maximization subject to non-truthful payment sematics is not trivial. See Appendix~\ref{app:unrevelation} for more details.}
	\item (non-truthful) winner-pays-bid and all-pay mechanisms,
          additive revenue approximation, and bounded and regular
          value distributions; and
	\item truthful mechanisms, multiplicative revenue
          approximation, and (unbounded) regular value distributions.\footnote{These results are more general than the results of \citet{DHP16}, who require downward closure.}
\end{itemize}
Regular distributions are ones that satisfy a natural convexity
property; details are given in \Cref{prelims}.

A key primitive in our results is the \emph{i.i.d.\ rank-based position
auction}, a model popularized in the study of ad auctions on search engines,
cf.\ \citet{JM-08}.  In such an auction agents are assigned to
positions, with higher positions having higher allocation
probabilities.  In an i.i.d.\ position auction the agents' values are
drawn from the same distribution.  Equilibria in these auctions are
unique and efficient \citep{CH13}, and are simple to compute using standard characterizations of Bayes-Nash equilibrium.  One way to view our results is as
a reduction from sample complexity in general single-parameter
environments with non-identically distributed agents to inference and design in
i.i.d.\ position auctions. 

To implement this reduction, our mechanisms use run-time
samples from their own bid distribution.  Doing so effectively
replaces competition between agents from distinct distributions with
competition between agents with identical distributions (from the
run-time samples). We show that to any agent, this is strategically equivalent to bidding in an i.i.d.\ position auction. Our mechanism therefore inherits the simple equilibrium structure of the latter mechanisms. We further show that mechanisms
in our family exist that are approximately optimal, i.e., the representation
error of our family is small, and we reduce the problem of analyzing the
generalization error to the problem of estimating appropriate expected order
statistics of design-time samples from bids in any i.i.d.\ position
auction with the same agent distributions.  For i.i.d.\ position
auctions with standard payment formats like winner-pays-bid and
all-pay, \citet{CHN16} solve this inference problem.  For the
truthful format, we give a straightforward solution and analysis.

\paragraph{Related Work} Typical welfare and revenue analyses of non-truthful mechanisms take standard auctions and analyze their welfare in worst-case equilibrium. Notable examples include \citet{ST13}, who prove that first-price and all-pay auctions have welfare within a constant fraction of optimal, and \citet{HHT14}, who derive a similar result for revenue. See \citet{RST17} for a complete survey. Unfortunately, except in settings such as i.i.d.\ position auctions, where equilibrium is efficient \citep{CH13}, asymmetries in the distributions and set of feasible allocations seem to guarantee a nontrivial fraction of welfare or revenue is lost in worst-case instances. In fact, \citet{DK15} show that standard analysis techniques cannot prove near-optimal revenue results for non-truthful mechanisms in a variety of common settings. Notably, these results only apply to designs and analyses which are agnostic to agents' distributions.

We consider mechanism design from data to circumvent these negative results. Design of truthful mechanisms from data has been studied extensively. For single-parameter mechanism design, this includes the learned finite support auctions of \citet{elk-07}, the learned monopoly reserve of \citet{DRY10}, and \citet{CR-14} as well as followups \citep{DHP16,RS16,MR15,GN-17,GHZ19}. A robust literature on the sample complexity of multi-parameter mechanism design also exists. As we only consider single-parameter agents, this literature is beyond the scope of our discussion. Relatively little work exists on the design of non-truthful mechanisms from data. A notable exception is \citet{CHN14}, whose inference methodology for i.i.d.\ position auctions are used in our own framework. All papers mentioned above use only design-time samples.

Our design approach has close connections to the work on reducing
Bayesian incentive compatible mechanism design to Bayesian algorithm
design. This work relies on both design-time and run-time samples. The
reductions of \citet{HL-15}, \citet{HKM11}, and \citet{BH11} show that
given an algorithm $\mathcal A$ and design-time sample access to
agents' value distributions, one can construct an $\epsilon$-Bayesian
incentive compatible mechanism with at most $\epsilon$ less expected
welfare than $\mathcal A$ in polynomial time. The constructed
mechanisms can be implemented without reliance on additional samples,
but fall short of exact incentive compatibility. For limited families
of preferences, the former two papers show how to use run-time samples
to achieve exact incentive compatibility, and \citet{DHKN-17} recently
showed how to use run-time samples to produce a fully general exactly
incentive compatible reduction. Achieving exact incentive
compatibility in blackbox reductions with only design-time samples is
an open problem for even single-parameter agents. Indeed, eliminating
the dependence on run-time samples in our setting would likely imply a
breakthrough in this regard.

\paragraph{Organization}
In \Cref{prelims}, we lay out notation and preliminary results. We
present our parametrized family of mechanisms, which we term
surrogate-ranking mechanisms, in \Cref{sec:rank-based}, and analyze
their equilibria. In \Cref{sec:BAR}, we show that for any set of value
distributions, there exists a surrogate-ranking mechanism that uses
polynomially many run-time samples and obtains nearly-optimal welfare
or revenue. Finally, in \Cref{sec:inference}, we show how to learn
such a mechanism using polynomially many design-time samples.

	\section{Preliminaries}
\label{prelims}

\newcommand{\feasibles}{{\mathcal X}}

This work considers the \emph{single-parameter} \emph{independent
  private value} model of mechanism design.  We describe this model in
{\em quantile} space where the geometry of approximation mechanisms is
more transparent \citep[cf.][]{har-13}.  There are $n$ agents drawn
independently and uniformly at random from $n$ populations.  Agents
are distinguished by their quantile with respect to their own
population.  The {\em quantile} $\quanti$ of agent $i$ is the measure
of population $i$ with higher values.  The {\em value function}
$\vali$ of population $i$ maps agent $i$'s quantile to her value as
$\vali(\quanti)$ and, with a uniformly drawn quantile, induces a value
distribution.  Profiles of agent values and quantiles are
denoted by $\vals = (\vali[1],\ldots,\vali[n])$ and $\quants =
(\quanti[1],\ldots,\quanti[n])$, respectively.

An {\em allocation} is $\allocs = (\alloci[1],\ldots,\alloci[n])$
where $\alloci \in\{0,1\}$ is an indicator for agent $i$ being served.
The space of feasible alloctions is given by $\feasibles \subset
\{0,1\}^n$.  (Notably, we do not require that $\feasibles$ be downward
closed.)  Agent $i$ can be assigned a non-negative payment denoted
$\pricei$ and her utility is linear in allocation and payment as
$\vali(\quanti)\,\alloci - \pricei$.

A mechanism takes as input a profile of bids $\bids =
(\bidi[1],\ldots,\bidi[n])$ and outputs a feasible allocation $\allocs
\in \feasibles$ and agent payments $\prices$.  A mechanism consists of
an \emph{allocation algorithm} $\ballocs(\bids)$, which maps bid
profiles to a feasible allocation, and a payment rule
$\bprices(\bids)$, which maps bid profiles to a non-negative payment
for each agent.  A standard allocation algorithm is {\em
  highest-bids-win} which is defined by $\ballocs(\bids) \in
\argmax_{\allocs \in \feasibles}\, \sum\nolimits_i\bidi\,\alloci$. We
consider payment rules defined directly from the allocation algorithm
according to standard payment formats.  The \emph{winner-pays-bid}
format has payment rule $\bpricei(\bids)=\bidi\, \balloci(\bids)$, and
the \emph{all-pay} format has payment rule $\bprice(\bids)=\bidi$.
Mechanisms with these payment formats do not have truth-telling as an
equilibrium.  The truthful payment format is defined according to the
payment identity (below, \Cref{thm:myerson}) and can be implemented as
an integral or with any of a number of unbiased estimators with
expectation equal to the integral (see, e.g., \citealp{HL-15}).

We analyze non-truthful mechanisms in Bayes-Nash equilibrium (BNE):
each agent's report to the mechanism is a best response to the
distribution of bids induced by other agents' strategies.  The
strategy of agent $i$ is denoted $\strati$ and maps the agent's
quantile to a bid and, with a uniformly drawn quantile, induces a bid
distribution.  The mechanism $(\ballocs,\bprices)$, the agents'
strategies $\strats$, and the distribution over quantiles induce
interim allocation and payment rules.  Agent $i$'s {\em interim
  allocation rule} is
$\alloci(\quanti)=\expect[\quantsmi]{\balloci(\strats(\quants))}$ and
{\em interim payment rule}
$\pricei(\quanti)=\expect[\quantsmi]{\bpricei(\strats(\quants))}$.
\citet{M81} characterized the interim allocation and payment rules
that arise in BNE when agents' values are independently distributed.

\begin{theorem}[\citealp{M81}]
\label{thm:myerson}
For independently distributed agents, interim allocation and payment
rules are induced by a Bayes-Nash equilibrium with onto strategies if
and only if for each agent $i$,
\begin{enumerate}
\item (monotonicity) 
\label{thmpart:monotone}
allocation rule $\alloci(\quant)$ is monotone
non-increasing in $\quanti$, and
\item 
\label{thmpart:payment}
(payment identity) payment rule $\pricei(\quanti)$ satisfies
$\pricei(\quanti) = \vali(\quanti)\,\alloci(\quanti) +
\int_{\quanti}^{1} \alloci(r) \, \vali'(r)\, d r + \pricei(1)$.
\end{enumerate}
\end{theorem} 

This paper studies the objectives of welfare and revenue.  The
\emph{welfare} of a mechanism is $\expect{\sum_i
  \vali(\quant)\,\alloci(\quanti)}$.  The optimal mechanism for
welfare allocates the value-maximizing feasible set, which is monotone
and therefore implementable with payments via
Theorem~\ref{thm:myerson}.  The \emph{revenue} of a mechanism is given
by $\expect{\sum_i \pricei(\quanti)}$.  The revenue of a mechanism is
easily analyzed in quantile space in terms of revenue curves and
marginal revenue as follows.
\begin{lemma}[\citealp{M81}; \citealp{BR89}]
\label{lem:vvals}
In BNE, the expected payment of an agent $i$ satisfies
$$\sexpect[\quanti]{\pricei(\quanti)} 
= 
\sexpect[\quanti]{-\alloci'(\quanti)\,\revi(\quanti)} + \revi(1)\,\alloci(1)
= 
\sexpect[\quanti]{\margi(\quanti)\,\alloci(\quanti)} + \revi(0)\,\alloci(0) 
$$
where the {\em revenue curve} $\revi(\quanti) = \vali(\quanti)\,\quanti$ gives the revenue from posting price $\quanti$ and the {\em marginal revenue} $\margi(\quanti) = \vali(\quanti) + \vali'(\quanti) \quanti$ is its derivative.  (Note that the derivatitives of the allocation rule $\alloci'(\cdot)$ and value function $\vali'(\cdot)$ are non-positive.)
\end{lemma}
The first equality follows from revenue equivalence and noting that
the allocation rule $\alloci$ is equivalent to offering a randomized
posted price with price distributed according to the density function
$-\alloci'(\cdot)$ with a pointmass of $\alloci(1)$ at price
$\vali(1)$. The second equality follows from integration by parts.
The optimal mechanism can be easily identified from the second
equality as maximizing the surplus of marginal revenue.  The value
distributions are called {\em regular} when the revenue curves are
concave, or equivalently, the marginal revenues are monotonically
non-increasing.  

In many environments of interest, the additive terms
$\revi(0)\,\alloci(0)$ and $\revi(1)\,\alloci(1)$ are zero.  For
example, when the strongest agent $\quanti=0$ in the population has
finite value $\vali(0)$, then the revenue when we post price
$\vali(0)$ is $\revi(0) = 0$ as only a zero measure of the population
will buy at such a price.  When the weakest agent in the
population $\quanti=1$ has value $\vali(1) = 0$ then $\revi(1) = 0$ as
the revenue from posting price 0 is zero. 


\paragraph{Position Auctions} 

{\em I.i.d.\ rank-by-bid position auctions} play a fundamental role in
our analysis.  In i.i.d.\@ environments the agents' value functions are
identical $\vali = \vali[j]$ for all agents $i$ and $j$.  An $n$-agent
position auction is defined by $n$ position weights
$\wali[1]\geq\ldots\geq \wali[n] \in [0,1]$ and an outcome is an assignment of
agents to positions.  If agent $i$ is assigned to position $j$ her
allocation is $\alloci = 1$ with probability $\wali[j]$ and zero
otherwise, i.e., $\expect{\alloci \given \text{agent $i$ is assigned
    slot $j$}} = \wali[j]$.  The rank-by-bid allocation algorithm
assigns agents to positions assortatively by bid.  The following
theorem shows that Bayes-Nash equilibria in rank-by-bid position
auctions are straightforward.

\begin{theorem}[\citealp{CH13}]
\label{thm:positionuniqueness}
In i.i.d.\ position environments, the rank-by-bid winner-pays-bid and
all-pay auctions have a unique and welfare-maximizing Bayes-Nash
equilibrium (in which agents are assigned to positions in order of
their true values), i.e., $\strati(\cdot) = \strati[j](\cdot)$ for all
agents $i$ and $j$.
\end{theorem}

	\section{Surrogate-Ranking Mechanisms}
\label{sec:rank-based}

In this section, we describe the parameterized family of
mechanisms for which we demonstrate polynomial sample complexity.  A
mechanism in this family has run-time sample access to the
equilibrium bid distribution of each agent.  An agent's bid can be
compared to these samples to estimate the agent's strength relative to
their value distribution.  The mechanism then allocates solely on the
basis of the agents' ranks.  The choice of parameters will determine
the exact mapping between ranks and allocations.  This approach can be
paired with any of the standard payment formats: winner-pays-bid,
all-pay, or truthful.


\label{sec:SAs}

\begin{definition}
\label{d:SRA}
A {\em surrogate-ranking mechanism} (SRM) is parameterized by $n\num$
{\em surrogate values} $\Surs$, with $\Suri =
\{\suri^1\geq\ldots\geq\suri^{\num}\}$ for each agent $i$.  The
input to the mechanism is a profile of bids.
\begin{enumerate}[1.]
\item A surrogate value is
calculated for each agent $i$ as:
\begin{enumerate}[(a)]
\item draw $\num-1$ run-time samples from the
agent's bid distribution, 
\item calculate the rank $\ranki$ of the agent's bid relative to
  these samples,
\item select the agent's surrogate value $\suri = \suri^{\ranki}$ according to the
  agent's sample rank.
\end{enumerate}
\item For space $\feasibles$ of feasible allocations, the algorithm allocates
to maximize the {\em surrogate surplus} $\argmax_{\allocs \in \feasibles}
\sum_i \suri\,\alloci$.
\item Payments are assigned according to any standard payment format,
  e.g., winner-pays-bid, all-pay, or truthful.
\end{enumerate}
\end{definition}

In the paper we will focus on surrogate ranking mechanisms as defined
above where the allocation is chosen to maximize the surrogate
surplus, i.e., $\sum_i \suri\,\alloci$.  Our methods extend in a
straightforward manner to settings where computing such an allocation
is intractable.  For approximation algorithms where surrogate
allocations is monotone in surrogate values, all analyses in this
paper hold with an additional multiplicative performance loss equal to
the approximation factor of the algorithm.  Non-monotone algorithms
can be made monotone via the methods of \citet{HL-15} or
\citet{HKM-15}.

Subsequently in \Cref{sec:inference}, we will show how to identify good
surrogate values from design-time samples.  The remainder of this
section is devoted to characterizing equilibria in surrogate ranking
mechanisms.

\subsection{Equilibria of Surrogate-Ranking Mechanisms}

\label{sec:eq}

We now analyze the equilibrium of winner-pays-bid and all-pay
surrogate-ranking mechanisms.\footnote{In Section~\ref{sec:inference},
  we also consider truthful SRMs. Equilibrium analysis for truthful
  SRMs is trivial.}  To do so, we first give a natural generalization
of Bayes-Nash equilibrium to mechanisms with run-time samples from
its own bid distribution.

\begin{definition} 
A {\em stationary equilibrium (with samples)} in a mechanism (with
samples) is a strategy profile $\strats$ where the strategy of each
agent is in best response to distribution of bids induced by the
strategies in the mechanism with sample access to the same bid
distributions.
\end{definition}

We will show that stationary equilibria of surrogate ranking mechanisms are easy to characterize. Each agent plays the unique Bayes-Nash equilibrium strategy of an i.i.d.\ position auction for their distribution in a position environment derived from the choice of surrogate values.
Specifically, rather than competing with
other agents in the mechanism, an agent $i$'s bid competes with other
bids from her bid distribution which gives an outcome which is
equivalent to the equilibrium of a position auction with
agents with values drawn only from population $i$, who share a distribution.

We begin by analyzing the distribution of assigned surrogate values in
a stationary equilibrium.  Recall, each agent's strategy $\strati$, on
a uniform quantile, induces a distribution over bids.  Notice that, in
the surrogate-ranking mechanism with sample access to this bid
distribution, the surrogate value assigned to $i$ will be uniformly
distributed from the set $\Suri$ of $i$'s surrogate values. This implies:

\begin{lemma}
\label{lem:uniformity}
In any stationary equilibrium of a surrogate ranking mechanism, and
any agent $i$ and surrogate value $\suri^j\in \Suri$, the ex ante
probability agent $i$ is assigned $\suri^j$ is $1/\num$.
\end{lemma}

\Cref{lem:uniformity} implies that in a stationary equilibrium,
the probability of allocation associated with a particular surrogate
value is fully determined by the other agents' sets of surrogate
values, and not by the form of the equilibrium bidding strategies, or
even by the agents' value distributions. This 
characterization of outcomes can be formalized as follows.

\begin{definition}
\label{def:charweights}
For each agent $i$, let $\Suri =
\{\suri^1\geq\ldots\geq\suri^{\num}\}$ be agent $i$'s set of surrogate
values, and let $\epallocs$ denote the surrogate surplus maximizing
allocation rule $\epallocs(\surs) = \max_{\allocs \in \feasibles}
\sum_i \suri\,\alloci$. The \emph{characteristic weights}
$\Wali = \{\wali^1 \geq \ldots \geq \wali^{\num}\}$ for agent $i$ are defined by
calculating the allocation probability associated with each surrogate
when the surrogates of other populations are drawn uniformly at
random, i.e., $\wali^j = \sexpect{ \epalloci (\suri^j,\sursmi)}$ for
each surrogate $j$ and uniform random $\sursmi$ from $\Sursmi$.
\end{definition}

We now show that from each agent's perspective, stationary equilibria
in surrogate-ranking mechanisms look like a position auction among
agents with the same value function.  These agents compete for the
characteristic weights of their population's surrogate values.  Under the
pay-your-bid or all-pay payment formats, they therefore inherit the
equilibrium of rank-based position auctions, which is shown by
\citet{CH13} to be efficient (i.e., to rank agents by values) and
unique.
\begin{theorem}
\label{thm:sruniqueness}
For any profile of value functions $\vals$, surrogate values $\Surs$, and
characteristic weights $\Wals$; the unique stationary equilibrium of
the winner-pays-bid (resp.\ all-pay) SRM is given by each agent $i$
bidding according to the unique and efficient BNE $\strati$ of the
i.i.d.\ winner-pays-bid (resp.\ all-pay) position auction with
position weights $\Wali$ and value function $\vali$.
\end{theorem}
\begin{proof}
Assume an arbitrary stationary equilibrium and consider an agent $i$.
By Lemma~\ref{lem:uniformity}, the stationary equilibrium induces a
uniform distribution over each other agent's assigned surrogate
values. It follows that if agent $i$ is assigned surrogate value
$\suri^j$, then they are allocated with probability
$\wali^j$. Moreover, since the surplus-maximizing allocation algorithm
is monotone, characteristic weights for agent $i$ are monotone as
well. Hence, placing the $j$th highest bid among the run-time samples
will cause $i$ to be assigned the $j$th highest characteristic
weight. Thus, agent $i$ faces the same bidding problem as if they
played in the equilibrium of the i.i.d.\ position auction with
position weights $\wali^1,\ldots,\wali^{\num}$ and value function
$\vali$.  Thus, agent $i$ bids according to the BNE of this
i.i.d.\ position auction.  This BNE is efficient, i.e., bids are in
the same order as values, and unique. 

Uniqueness of the stationary equilibrium follows by uniqueness of
characteristic weights under any stationary equilibrium
(\Cref{def:charweights}), which are determined only by the set of
surrogate values $\Surs$, and the uniqueness of symmetric Bayes-Nash equilibrium in i.i.d.\ position auctions, which follows from revenue equivalence.
\end{proof}

One consequence of \Cref{thm:sruniqueness} is that the bidding problem faced by agents is strategically simple. Given accurate estimates of the characteristic weights, the symmetric equilibrium bids for the corresponding position environment for each population can be computed using \Cref{thm:myerson}.

\subsection{Equivalence of Surrogate Ranking Mechanisms}

Surrogate ranking mechanisms are equivalent for revenue and welfare,
irrespective of their payment format.  It is helpful to relate this
equivalence to the famous revenue equivalence result of \citet{M81}.
In the latter, two mechanisms with the same equilibrium outcome (and
with the same expected payment of the agent with the lowest value in
the support of the distribution, usually zero) have the same expected
revenue.  For instance, with i.i.d.\ distributions the single-item
first-price and second-price auctions have the same equilibrium outcome, i.e.,
the highest valued agent wins, and thus, by revenue equivalence, the
same expected revenue.  With non-identical distributions, these
auctions do not have the same equilibrium outcome and, thus, do not
generally have the same expected revenue.  Our equivalence result, in
contrast, holds for surrogate-ranking mechanisms in asymmetric
environments (distributions and feasibility constraints).

\begin{theorem}
\label{thm:srm-revelation}
For any fixed surrogate values and value functions, the expected
welfare (resp.\ revenue) of the winner-pays-bid, all-pay, and truthful
surrogate-ranking mechanisms in stationary equilibrium with samples
are equal.
\end{theorem}

\begin{proof}
This theorem follows because each agent is playing the symmetric BNE of an i.i.d.\ position auction with characteristic weights that are independent of the payment format. Such BNE are welfare-equivalent for each agent. Welfare equivalence implies, by the usual argument, revenue
equivalence.
\end{proof}

\Cref{thm:srm-revelation} gives a revelation principle for surrogate
ranking mechanisms.  Bounds on the revenue and welfare of the truthful
surrogate ranking mechanism implies the same bounds on that of the
non-truthful ones because their equilibrium outcomes are the same in
expectation.

	\section{Representation Error of Surrogate Ranking Mechanisms}
\label{sec:BAR}

In this section, we show that the representation error of surrogate-ranking mechanisms is small. In other words, there always exists some choice of surrogate values which induces near-optimal expected welfare or revenue. Hence, SRMs satisfy condition \Cref{cond:performance} in our statement of the problem of non-truthful sample complexity. More precisely, we prove the following guarantee:
\begin{theorem}
\label{thm:welfare}
\label{THM:WELARE}
There exists a surrogate-ranking mechanism with winner-pays-bid, all-pay, or
truthful payment semantics which attains a $(1-
O(\sqrt[3]{n/\num}))$-fraction of the optimal welfare in stationary equilibrium. With regular distributions, there exists such a mechanism which attains a $(1-
O(\sqrt[3]{n/\num}))$-fraction of the optimal revenue in stationary
equilibrium.
\end{theorem}
In \Cref{sec:inference}, we show how to estimate a SRM which is nearly welfare- or revenue-optimal among all mechanisms in the family. This mechanism consequently inherits the guarantees of \Cref{THM:WELARE}, up to error from estimation.

Before discussing the proof of \Cref{THM:WELARE} we observe that a less general result
follows from a theorem of \citet{HKM11}. In this paper, the authors
consider an allocation procedure which can be interpreted as a surrogate ranking
mechanism. They show that for agents whose values are
distributed on the interval $[0,1]$, their mechanism has an additive welfare loss of at most $n/(4\sqrt T)$. 
We improve on this welfare guarantee in three ways. First, we derive a guarantee for arbitrary distributions, even those with unbounded support. Second, our bounds will be multiplicative. Finally, our bounds apply to revenue in addition to welfare.\footnote{With the additional assumption that virtual values are bounded below by $-\underline \phi$, the guarantee of \citet{HKM11} applies to revenue as well, with an additional factor of $1+\underline \phi$ applied to the loss.}

To prove \Cref{THM:WELARE} we define a family of mechanisms, \emph{surrogate binning mechanisms}, which coarsen quantile space into uniform bins, assign agents to surrogate values based on their bin, and maximize surplus with respect to the assigned surrogate values. Binning mechanisms interpolate between the structure of the optimal mechanism and SRMs; they treat agents coarsely based on their location in their distribution, whereas the optimal mechanism treats agents based on their exact location in their distribution, and SRMs treat agents coarsely based on their rank among runtime samples. \Cref{THM:WELARE} follows from identifying a set of surrogate values for which the surrogate binning mechanism is nearly optimal and the SRM is not much worse.

We formally define surrogate binning mechanisms in \Cref{SEC:BIN}, and show that for any set of surrogate values where each agent's highest $k$ and lowest $k$ surrogate values are each identical, surrogate ranking and binning mechanisms perform comparably, with the loss depending on $k$. It therefore suffices to find a set of surrogate values satisfying this property that yields a near-optimal surrogate binning mechanism. We do so in \Cref{SEC:BINAPX}.

\subsection{Surrogate-Binning Mechanisms}
\label{SEC:BIN}

We now discuss our intermediate family of mechanisms, surrogate-binning mechanisms. Recall that surrogate-ranking mechanisms assign agents surrogate values from a fixed set based on their bid's rank among run-time samples, and allocates the feasible set with the highest total surrogate value. Surrogate-binning mechanisms perform the same procedure, but assign surrogate values based on a coarsening of each distribution's quantile space, rather than ranking among samples. Formally:

\begin{definition}
\label{d:SRBA}
A {\em surrogate binning mechanism} (SBM) is parameterized by $n\num$
{\em surrogate values} $\Surs$, with $\Suri =
\{\suri^1\geq\ldots\geq\suri^{\num}\}$ for each agent $i$.  The
input to the mechanism is a profile of quantiles for each agent.
\begin{enumerate}[1.]
\item A surrogate value is
calculated for each agent $i$ as $\suri=\suri^{j}$, where $\quanti\in ((j-1)/\num,j/\num]$.
\item For space $\feasibles$ of feasible allocations, the algorithm allocates
to maximize the {\em surrogate surplus} $\argmax_{\allocs \in \feasibles}
\sum_i \suri\,\alloci$.
\item Charge truthful payments.
\end{enumerate}
\end{definition}

We refer to the interval $((j-1)/\num,j/\num]$ in agent $i$'s quantile space as the \emph{$j$th bin} for that agent. SRMs and SBMs are similar in their use of surrogate values, but differ in that SRMs use ranking to discriminate between high- and low-valued agents, while SBMs compare agents' bins. In this section, we show that as long as neither mechanism discriminates at the top or bottom of the value distribution, i.e. the highest surrogate values for each agent are the same and the lowest surrogate values for each agent are the same, then binning and ranking discriminate comparably well. Formally:

\begin{theorem}
\label{thm:bvr}
For surrogate values
$\suri^1\geq \suri^2\geq\ldots\geq \suri^T$ with
$\suri^1=\suri^2=\ldots=\suri^{\underline k}$ and
$\suri^{\overline k}=\suri^{\overline k+1}=\ldots=\suri^T$ for each
population $i$, the SRM for $\Surs$ with all-pay, winner-pays-bid, or truthful payments in stationary equilibrium attains a
$(1-O(\underline k^{-1/2}))$-fraction of the welfare of
the SBM for $\Surs$. If distributions are regular, then the
SRM attains a $(1-O(\min(\underline k,T-\overline k)^{-1/2}))$-fraction of the SBM's virtual
surplus as well.
\end{theorem}

To derive \Cref{thm:bvr}, first prove in \Cref{SEC:RVP} that the simplest ranking mechanism, a $k$-unit, highest-bids-win auction, approximates the performance of posted pricing, which can be seen as the simplest binning mechanism. In \Cref{SEC:RVB}, we then show that SRMs and SBMs appear as distributions over these simpler single-agent mechanisms, which implies \Cref{thm:bvr}.

\subsubsection{Ranking Versus Pricing}
\label{SEC:RVP}

In this section, we establish a lemma relating the performance of simple mechanisms which allocate agents based on their rank among samples and those that allocate based on their quantiles. This result will serve as a technical building block for comparing surrogate ranking and surrogate binning mechanisms.

Consider the following two methods for selling an unlimited supply of items to $\num$ identically distributed agents. First, the seller could impose an {\em ex post} supply limit of $k$ units, and sell them to the $k$ highest-valued agents. Second, the seller could instead relax to a $k$-unit supply limit {\em in expectation} by selling to all agents with quantile below $k/\num$ (i.e. posting a price of $\val(k/\num)$). Note that each approach can be implemented as a surrogate-ranking or surrogate-binning mechanism, respectively. Formally:

\begin{definition}
The {\em $k/\num$-price posting mechanism} allocates agents if and only if their quantile is below $k/\num$, for some integer $k$. This can be achieved by posting the price with quantile $k/\num$ in an environment with unlimited supply.
\end{definition}

\begin{definition}
The {\em top $k$-of-$\num$ mechanism} for $\num$ agents ranks agents by value and allocates the $k$ agents with the highest values. It charges winners the $k+1$st highest value.
\end{definition}

As the law of large numbers might suggest, these two mechanisms
perform comparably for large $\num$ when $k$ is bounded away from the
extremes ($1$ and $\num-1$).

\begin{lemma}
	\label{lem:rankvprice}
	\label{LEM:RANKVPRICE}
	The top $k$-of-$\num$ mechanism attains at least a $(1-O(k^{-1/2}))$-fraction of the welfare of the $k/\num$-price posting algorithm with $\num$ i.i.d. agents. If values are regularly distributed, then it attains at least a $(1-O(\min(k,T-k)^{-1/2}))$-fraction of the revenue of the $k/\num$-price posting mechanism.
\end{lemma}

A proof can be found in \Cref{APP:RVP}. The result follows from an explicit characterization of the distributions for which the ratio in performance between the two mechanisms is largest. Hence, the bounds of the lemma are tight.

One immediate generalization of \Cref{LEM:RANKVPRICE} is to distributions over pricing mechanisms and analogous distributions of top-$k$ mechanisms.

\begin{lemma}
\label{lem:distprice}
Consider a distribution over $k/\num$-price posting mechanisms for $\num$ i.i.d. agents, where the highest price is at quantile $\underline k/\num$ and the lowest price is at quantile $\overline k/\num$. The same distribution over corresponding top-$k$-of-$\num$ mechanisms attains a $(1-O({\underline k}^{-1/2}))$-fraction of the welfare of the distribution over price-posting mechanisms. If values are regularly distributed, then the distribution over top-$k$-of-$\num$ algorithms attains a $(1-O(\min(\underline k,T-\overline k)^{-1/2}))$-fraction of the price-posting revenue as well.
\end{lemma}

\begin{proof}
Lemma~\ref{lem:rankvprice} implies that for each price in the distribution of the price-posting mechanism, there is a top-$k$-of-$\num$ algorithm which approximates it and which appears with the same probability. The approximation ratio of a distribution over pairwise approximations is at least the approximation from the worst pair. 
\end{proof}

\subsubsection{Ranking Versus Binning}
\label{SEC:RVB}

Lemma~\ref{lem:distprice} shows that in simple settings, ranking
discriminates almost as effectively as pricing. We now extend this idea to surrogate-ranking and surrogate-binning mechanisms.

Any allocation rule can be viewed by agents as a distribution over posted prices. For surrogate binning mechanisms, this distribution is especially easy to analyze. Because  surrogate-binning mechanisms treat agents solely based on the bin into which their quantile falls, and because bins are distributed evenly in quantile space, it follows that for surrogate-binning mechanisms, the allocation rule is piecewise-constant, with break points occurring at multiples of $1/\num$. It follows that the corresponding distribution over posted prices can actually be viewed as a distribution over $k/\num$-quantile price posting mechanisms. Moreover, since quantiles (and agents' bins) are distributed uniformly, the probability of allocation associated with a particular surrogate value is exactly that surrogate value's characteristic weight. The probability of the price with quantile $j/\num$ being offered is therefore the marginal characteristic weight, $w_i^j-w_i^{j+1}$.

Meanwhile, a standard fact from the study of position auctions is that any rank-based position auction can be represented as a convex combination of $k$-unit auctions. Since the allocation rule agents face in the stationary equilibrium of a surrogate-ranking mechanism is identical to that of a rank-based position auction, a similar analysis applies. We summarize the above discussion with the following lemma:

\begin{lemma}
\label{lem:bin}
Any surrogate-binning (resp. surrogate-ranking) mechanism with surrogate values $\{\psi_i^j\}^{j=1,\ldots,T}_{i=1,\ldots,n}$ appears to each agent $i$ as a distribution over price-posting (resp. top-$k$) mechanisms. The probability of offering the price with quantile $\tfrac{j}{T}$ (resp. of allocating $j$ units) is given by $w_i^{j}-w_i^{j+1}$, where $w_i^0 = 1$, $w_i^{T+1}=0$, and $w_i^j$ is the characteristic weight for $\psi_i^j$ for $j=1,\ldots, T$.
\end{lemma}

Notice that if
$\psi_i^1=\ldots=\psi_i^{\underline k}$ for some positive integer $\underline k$, the binning mechanism's
allocation rule on the first $\underline k$ intervals of distribution
$i$'s quantile space will be constant. If $\psi_i^{\overline k}=\ldots=\psi_i^T$ for some positive integer $\overline k$, then the allocation rule on the last $T-\overline k$ intervals will be constant. In terms of distribution $i$'s
randomization over posted pricings, the highest nontrivial price
offered has quantile $\underline k/\num$, and the lowest has quantile
$\overline k/\num$. These extremal quantiles drive the approximation
guarantees relating pricing to ranking. Consequently, we obtain the following theorem:

\begin{theorem}
\label{thm:bvr}
For surrogate values
$\psi_i^1\geq \psi_i^2\geq\ldots\geq \psi_i^T$ with
$\psi_i^1=\psi_i^2=\ldots=\psi_i^{\underline k}$ and
$\psi_i^{\overline k}=\psi_i^{\overline k+1}=\ldots=\psi_i^T$ for each
population $i$, the surrogate ranking mechanism attains a
$(1-O(\underline k^{-1/2}))$-fraction of the welfare of
the binning mechanism. If distributions are regular, then the
surrogate ranking mechansim attains a $(1-O(\min(\underline k,T-\overline k)^{-1/2}))$-fraction of the binning mechanism's virtual
surplus.
\end{theorem}


\subsection{Approximately Optimal Binning Mechanisms}
\label{SEC:BINAPX}

Surrogate binning mechanisms are parametrized by their surrogate values. In this section, we show how to choose a set of surrogate values that yields an SBM which (1) attains near-optimal welfare or revenue, and (2) permits approximation by a surrogate ranking mechanism via \Cref{thm:bvr}. More precisely:

\begin{theorem}
\label{THM:BIN}
For every $k\in\{1,\ldots, \num/2\},$ there exists a choice of surrogate values which yields a surrogate binning mechanism such that:
\begin{itemize}
    \item The welfare of the corresponding surrogate binning mechanism is at least a $(1-O(1/k))(1-O(\numag k/\num))$-fraction of the optimal welfare.
    \item For each agent, the highest $k$ surrogate values are identical and the lowest $k$ surrogate values are identical.
\end{itemize}
If value distributions are regular, an identical result holds for the objective of revenue.
\end{theorem}

Note that combining \Cref{thm:bvr} with \Cref{THM:BIN} and choosing $k=(n/\num)^{2/3}$ immediately implies the representation error bound of \Cref{THM:WELARE}. We devote this section to deriving surrogate values which prove \Cref{THM:BIN}. The high-level approach will be resampling. To illustrate the idea, consider mapping the $j$th surrogate of
agent $i$, which corresponds to an agent with quantile in $[\tfrac{j-1}{T},\tfrac{j}{T}]$
 to a redrawn value from their value distribution conditioned to this interval.  Such resampling does not change the induced
allocation rule for any other agents, and replaces the allocation rule for agent $i$
on their $j$th quantile interval with its average.

This na\"ive resampling procedure does not directly lead to an approximately optimal mechanism. For an example where this fails, note that for the top quantile interval, the optimal mechanism's allocation
probability at the very top of the interval may be much higher than its
average across the interval, while the highest values on the interval may
be much higher than the interval's average.  For example, if the value
and allocation rule are both $1$ for an $\epsilon$ measure and zero
otherwise, then the optimal welfare is $\epsilon$ and the welfare
from resampling is $\epsilon^2$.  A second issue is that we wish to be
able to apply \Cref{thm:bvr} to get a good approximation bound, and therefore need the $k$ highest and $k$ lowest surrogate values to each be the same for each agent. Resampling na\"ively for every bin does not achieve this objective.

To resolve both these issues, we will artificially inflate the values of agents with low quantiles, treating them as if they had the highest value in the support of their distributions, and deflate the values of high quantiles, treating these agents as if they had the lowest value in the support of their distributions. We formalize this approach in \Cref{SEC:BUFFER} and \Cref{SEC:RESAMPLE}. \Cref{SEC:BUFFER} first treats the process of exaggerating extreme quantiles in isolation, analyzing a procedure that can be applied to any mechanism. In \Cref{SEC:RESAMPLE}, we combine this with the resampling procedure discussed above to produce a surrogate binning mechanism satisfying the conditions of \Cref{THM:BIN}.

\subsubsection{Extremal Buffering}
\label{SEC:BUFFER}

In this section, we design and analyze a procedure for privileging agents with low quantiles and penalizing agents with high quantiles.
 We will show that our procedure does not significantly harm welfare or revenue, and in the next section, we will use this analysis as a foundation for proving \Cref{THM:BIN}. Our procedure will take as input an arbitrary truthful mechanism (e.g. a revenue-optimal mechanism), and modify its allocation rule. We will compare the virtual surplus of the new rule to the original mechanism. In the next section, we will combine this with a resampling procedure to obtain a surrogate binning mechanism.
 
Formally, consider an arbitrary truthful mechanism with allocation rule $\allocs=(\alloci[1],\ldots,\alloci[n])$. We will define a new allocation rule by remapping agents' quantiles to exaggerate extreme quantiles before applying $\allocs$. In particular, we consider the following transformation of $\allocs$:

\begin{definition}
	\label{def:hedge}
	Given a monotone allocation rule $\allocs$ and a quantile $q\in[0,1]$, the \emph{$q$-buffering rule for
	$\allocs$} runs $\allocs$ on agents with quantiles
	transformed for each agent as follows:
	\begin{itemize}
		\item For any $q_i\in[0,q]$, return 0.
		\item For any $q_i\in[q,1-q]$, return $(q_i-q)/(1-2q)$.
		\item For any $q_i\in [1-q,1]$, return $1$.
	\end{itemize}
\end{definition}

The $q$-buffering rule for $\allocs$ treats agents at the top of the distribution as if they had quantile $0$ and agents at the bottom as if they had quantile $1$. Moreover, for any agent $i$, conditioned on $\quanti\in[q,1-q]$, the distribution of remapped quantiles for $i$ is uniform. We show that the $q$-buffering procedure approximately preserves both welfare and revenue, as well as any other virtual surplus quantity that satisfies two basic properties satisfied by values or regular virtual values.
We will only consider $q$-buffering as an analysis tool, so it will suffice to analyze the virtual surplus generated by the resulting allocation rules independent of incentives. Note that the $q$-buffering procedure preserves monotonicity, so truthful payments could be found if desired. 
Formally: 

\begin{lemma}
	\label{lem:toppromotion}
	\label{LEM:TOPPROMOTION}
	For each agent $i$, let $\phi_i:[0,1]\rightarrow \mathbb R$ be an arbitrary nonincreasing virtual value function satisfying $\int_0^1 \phi_i(q)\,dq\geq 0$. The $q$-buffering rule for $\allocs$ attains at least a $(1-\frac{q}{1-q})(1-q)(1-2(n-1)q)$-fraction
	of the virtual surplus of $\allocs$.
\end{lemma}

The proof of Lemma~\ref{LEM:TOPPROMOTION} can be found in \Cref{APP:BUFFER}. Note that choosing $q=k/\num$ for some $k\in\{1,\ldots, \num/2\}$ yields an approximation factor of $(1-O(1/k))(1-O(\numag k/\num))$. The proof of the lemma follows from several observations about the similarity of $\allocs$ and the $q$-buffering rule for $\allocs$. One particular step in this analysis may be of
independent interest.  Recall the characterization of expected revenue
in terms of revenue curves and marginal revenue from \Cref{lem:vvals}:
\begin{align}
\label{eq:virtualsurplus}
\sexpect[\quant]{\price(\quant)} 
&= 
\sexpect[\quant]{-\alloc'(\quant)\,\rev(\quant)} + \rev(1)\,\alloc(1)
= 
\sexpect[\quant]{\marg(\quant)\,\alloc(\quant)} + \rev(0)\,\alloc(0).
\end{align} 
The first equality enables a geometric understanding of revenue.
Given an fixed allocation rule $\alloc$ in quantile
space, for two value functions $\vali[1]$ and $\vali[2]$ where
$\revi(\quant) = \quant\,\vali(\quant)$ satisfies $\revi[1](\quant)
\geq \revi[2](\quant)$, then the revenue from $\vali[1]$ on $\alloc$
is at least the revenue of $\vali[2]$ on $\alloc$.  This follows from
the first equality of equation \eqref{eq:virtualsurplus}, where the
expressions for revenue of both value functions are weighted integrals
over $\quant\in[0,1]$ with non-negative weights $-\alloci'(\quant)$.
Approximation bounds hold as well; specifically, if $\revi[2]$
approximates $\revi[1]$ at all $\quant\in[0,1]$ then the same
approximation holds for the revenue of any fixed allocation rule
$\alloc$. Note however that a similar result, with a fixed distribution and similar allocation rules is not implied by the second equation, as
the weights $\marg(\quant)$ are not generally all the same sign.  Instead, the
following lemma 
shows that two allocation rules with \emph{inverses} that are
approximately close have approximately the same revenue.  We state the
result for general virtual value functions $\virt(\cdot)$ and
cumulative virtual value curves $\cumvirt(\quant) = \int_0^\quant
\virt(r)\,dr$.  The assumption in the lemma on the cumulative virtual value
curve is that lines from the origin pass from below to above the
curve.  This assumption, for example, is satisfied by any revenue
curve, and it does not require regularity.

\begin{lemma}
\label{lem:apx}
For virtual value function $\virt(\cdot)$ and cumulative virtual value
$\cumvirt(q)=\int_0^{\quant} \virt(r)\,dr$ satisfying $\cumvirt(\alpha
\,\quant)\geq \alpha\, \cumvirt(\quant)$ for all quantiles $\quant$
and $\alpha \in [0,1]$, and any two allocation rules $\alloc_1$ and
$\alloc_2$ that satisfy $\alloc_1^{-1}(z)\geq \alloc_2^{-1}(z)\geq
\tfrac{1}{\alpha}\alloc_1^{-1}(z)$, the expected virtual surpluses satisfy
	\begin{align*}
	\expect[\quant]{\virt(\quant)\,\alloc_2(\quant)} + \cumvirt(0)\,\alloc_2(0)
   &\geq \tfrac{1}{\alpha} [\expect[\quant]{\virt(\quant)\,\alloc_1(\quant)} + \cumvirt(0)\,\alloc_1(0)]. 
	\end{align*} 
\end{lemma}

\begin{proof}
The virtual surplus of any allocation rule $\alloc$ can be rewritten
as $\int_0^1\virt(\quant)\,\alloc(q)\,dq + \cumvirt(0)\,\alloc(0) =
\int_0^1 \cumvirt(\alloc^{-1}(z))\,dz$. This follows by the first
equality of equation \eqref{eq:virtualsurplus} and a change of
variables to integrate the vertical axis rather than the horizontal
axis as follows:
\begin{align}
\notag
\int_0^1 {-\alloc'(\quant)\,\cumvirt(\quant) \, d\quant} + \cumvirt(1)\,\alloc(1) 
	&= \int_{\alloc(1)}^{\alloc(0)} \cumvirt(x^{-1}(z))\, dz + \int_{0}^{\alloc(1)} \cumvirt(1)\, dz\\
\label{eq:inverse}
	&= \int_{0}^{\alloc(0)} \cumvirt(\alloc^{-1}(z))\, dz.
\end{align}
Notice that the second line follows from the first line because $\alloc^{-1}(z) = 1$ for $z \in [0,\alloc(1)]$.
	
Now consider two arbitrary quantiles $\quant_1$ and $\quant_2$
satisfying $\tfrac{1}{\alpha}\quant_1\leq \quant_2\leq \quant_1$. By
assumption, we have $\cumvirt(\quant_2)\geq
\quant_2\cumvirt(\quant_1)/\quant_1\geq
\tfrac{1}{\alpha}\cumvirt(\quant_1)$. The assumption on the
approximation of the two allocation rules, namely $\alloc^{-1}(z)\geq
\ealloc^{-1}(z)\geq \tfrac{1}{\alpha}\alloc^{-1}(z)$ for all
$z\in[0,1]$, and the expected virtual surplus written as rewritten in
equation~\eqref{eq:inverse} both both $\alloc_1$ and $\alloc_2$, then,
suffice to prove the lemma.
\end{proof} 

\subsubsection{Approximately Optimal Reasmpling}
\label{SEC:RESAMPLE}

In the previous section, we analyzed the effect of exaggerating extreme quantiles. We now introduce an additional resampling transformation, and show that the performance loss is not significant when combined with the buffering procedure. Furthermore, the mechanism produced will be a surrogate binning mechanism. This will imply Theorem~\ref{THM:BIN}. 

As we did in the previous section, we will argue for an arbitrary monotone allocation rule $\allocs.$ Let $k\in\{1,\ldots,\num/2\}$ be given, and let $\hat \allocs$ denote the  $k/\num$-buffering rule for $\allocs$. Our resampling procedure is defined as follows:

\begin{definition}
    Given an allocation rule $\allocs$, the {\em resampling rule for $\allocs$} allocates according to the following randomized procedure:
\begin{itemize}
	\item For each agent $i$:
    \begin{itemize}
    	\item Compute $j$ such that $i$'s quantile $q_i\in[(j-1)/T,j/T]$.
        \item Resample a quantile uniformly from $[(j-1)/T,j/T]$.
 	\end{itemize}
 
  \item Run $\allocs$ on the resampled quantiles.
\end{itemize}
\end{definition}

This algorithm formalizes the procedure discussed at the beginning of \Cref{SEC:BINAPX}. Because of the problems discussed earlier, the resampling procedure alone does not necessarily preserve the performance of the original allocation rule. However, composing the resampling procedure with the buffering procedure of the previous section eliminates these pathologies, yielding:

\begin{lemma}\label{lem:resampling}
		For each agent $i$, let $\phi_i(q_i)$ be a nonincreasing virtual value function, and let $\mathbf x$ be an a monotone allocation algorithm. Further, for some $k\leq T/2$, let $\mathbf{\hat x}$ denote the $k/T$-buffering algorithm for $\mathbf x$. Then the resampling algorithm for $\mathbf{\hat x}$ obtains at least a $\frac{k}{k+1}(1-\frac{q}{1-q})(1-q)(1-2(n-1)q)$ the expected virtual surplus under $\mathbf x$, with $q=k/T$.
\end{lemma}

A proof can be found in \Cref{APP:RESAMPLE}. The main observations are that resampling the buffered allocation rule does not change the allocation of agents with extreme quantiles, and that the nonincreasing nature of $\virti$ implies that the loss from resampling moderate quantiles is small.

\begin{proof}[Proof of \Cref{THM:BIN}]
We argue for the objective of revenue. The argument for welfare is essentially identical. Let $\allocs$ be the allocation rule of the revenue-optimal mechanism, let $\hat{\allocs}$ be the $k/\num$-buffering rule based on $\allocs$. We first argue that the resampling rule for $\hat{\allocs}$ can be implemented as a surrogate binning mechanism with random surrogate values. In particular, the $j$th surrogate value for agent $i$, $\suri^j$, is produced in in the following way:
\begin{itemize}
    \item If $j\in \{1,\ldots, k\}$, set $\suri^j=\revi'(0)$
    \item If $j\in \{\num-k+1,\ldots, \num\}$, set $\suri^j=\revi'(1)$
    \item Otherwise:
    \begin{itemize}
        \item Resample $\quanti'$ uniformly from $((j-1)/(\num-2k),j/(\num-2k)]$.
        \item Set $\suri^j=\revi'(\quanti')$.
    \end{itemize}
\end{itemize}

These surrogate values satisfy the properties required for \Cref{THM:BIN}, but are randomized. Note that our performance guarantees are in expectation over both the random choice of surrogate values and the quantiles of agents. This implies the existence of a deterministic choice of surrogate values for which the performance guarantees hold in expectation over the quantiles of agents. The result follows. 
\end{proof}

	\section{Reduction from Sample Complexity to Rank-Based Inference}
\label{sec:inference}

We have shown that surrogate-ranking mechanisms possess a unique
stationary equilibrium (\Cref{thm:sruniqueness}), and that this
equilibrium may be analyzed as if it was truthful
(\Cref{thm:srm-revelation}). In this section, we show how to use bid
data to design a surrogate-ranking mechanism with near-optimal welfare
or revenue in stationary equilibrium.  Specifically, we reduce this design
problem to an inference problem which is better-understood: estimating
expected order statistics from bid data in i.i.d.\ position auctions.
This inference problem was solved by \citet{CHN16} for first-price
and all-pay mechanisms and is straighforward for truthful mechanisms.

Before giving details, we describe the high-level approach.  Recall
from \Cref{sec:BAR} that, for revenue with regular distributions or
welfare with general distributions, polynomially many surrogate values
$\num$ per agent suffice to obtain a $(1-\epsilon)$-fraction of the
optimal revenue or welfare via a surrogate-ranking mechanism.
Consider a surrogate-ranking mechanism with $T$ surrogate values per
agent. We first show in \Cref{sec:optimal} that the revenue- or
welfare-optimal choice of these $n\num$ surrogate values requires only
knowledge of the expected order statistics of the value or virtual value distribution. In
\Cref{sec:prop} we observe that this design approach is robust to error
from inference: if one uses imperfect estimates of expected order
statistics to design a SRM, then estimation error will propagate
cleanly to revenue or welfare loss.  Composing these three observations
above yields the desired reduction.


\Cref{sec:instantiation} concludes by instantiating the reduction with
an estimator for the requisite order statistics.  Specifically,
\citet{CHN16} show how to estimate expected order statistics using bid
data from all-pay and winner-pays-bid position auctions with bounded
distributions, and we show in \Cref{APP:SAMPLES} how to estimate
expected order statistics with truthful data from unbounded regular
distributions. These results imply polynomial sample complexity for
winner-pays-bid, all-pay, and truthful mechanism design.

\subsection{Optimal Surrogate-Ranking Mechanisms}
\label{sec:optimal}

Surrogate-ranking mechanisms are parametrized by $n\num$ surrogate
values. Each choice of surrogate values induces a different allocation
rule in stationary equilibrium, but by \Cref{thm:srm-revelation}, this
equilibrium allocation rule is the same under any of the standard
payment formats.  Hence the optimal surrogate values, by revenue equivalence, do not depend on the payment format and we may as well identify the optimal surrogate values for the truthful payment format. In this section, we characterize the
welfare- and revenue-optimal choices of surrogate values. To choose
our surrogate values optimally, we consider a relaxed problem of
maximizing a generic virtual surplus quantity subject to the
constraint that the allocation rule depend only on each agent's rank
among $\num-1$ other truthful run-time samples.  The solution to this
problem is straightforward given the  observation that if the
only information we have to make decisions on is the rank of an agent
against samples from her value distribution, then decisions should be made
based on expected order statistics.  The right choice of surrogate
values is the expected order statistics of the quantity of interest
for the objective, i.e., for welfare maximization, it's order statistics
for the value distribution for revenue maximization it's order
statistics for the distribution of marginal revenues (via the
characterization of expected revenue in \Cref{lem:vvals}). Formally, given $\num-1$ sampled quantiles for each agent, let $r_i$ denote the rank of the quantile $q_i$ of agent $i$ among these samples. We have:

\begin{theorem}
\label{thm:optimalSRM}
\label{THM:OPTIMALSRM}
The welfare-optimal surrogate-ranking mechanism uses surrogate values given by
$\suri^j=\expect[\quanti]{\vali(\quanti) \given \ranki=j}$. For
regular distributions, the revenue-optimal surrogate-ranking mechanism
uses surrogate values $\suri^j=\expect[\quanti]{\margi(\quanti) \given
  \ranki=j}$ where $\margi(\quanti) = \vali(\quanti) +
\quanti\,\vali'(\quanti)$ is the marginal revenue of agent $i$ at
quantile $\quanti$.
\end{theorem}

A formal
proof of this theorem is in \Cref{app:optimal-srm}. Because the welfare- and revenue-optimal surrogate ranking mechanisms are at least as good as any other surrogate ranking mechanism, it follows that they inherit the welfare and revenue guarantees of any other such mechanism. In particular, we obtain the following corollary to \Cref{thm:welfare}:
\begin{corollary}
The welfare- (resp.\ revenue-) optimal surrogate ranking mechanism
obtains a $(1- O(\sqrt[3]{n/T}))$-fraction of the optimal
welfare (resp.\ a $(1- O(\sqrt[3]{n/T}))$-fraction of the
optimal revenue with regular distributions) in stationary equilibrium.
\end{corollary}
\subsection{Propagation of Error}
\label{sec:prop}

The optimal surrogate ranking mechanism for welfare (resp. revenue)
uses surrogate values equal to the expected order statistics of each
agent's value (resp. virtual value) distribution. We now show that a
designer can in fact use noisy estimates of these quantities, and the
performance will degrade smoothly with the estimation error. As before, we
present our result for an arbitrary monotonic virtual value function
$\virti$ for each agent.

For monotonic $\virti$, though the expected order statistics are
monotone, estimates of these expected order statistics may not be.  However, if
the estimates of agent $i$'s expected order statistics are within an agent-specific error $\errori$ of being correct, then any natural method or making these estimates monotone will be similarly close, e.g., using the $j$th
estimate of $\max_{j' \leq j} \esuri^j$ instead of $\esuri^j$.  Of
course, the recommended method is the standard approach of ironing,
which for quantities like order statistics is formally described in
\citet{DHY15} and \citet{CHN16}.  In fact ironing is equivalent in
this setting to isotonic regression.  An advantage of ironing is that
it is the correct approach when the original $\virti$ is non-monotonic
and, omitting the details and consequences, our results for revenue
can be extended to {\em tail regular} distributions using this
approach, cf.\ \citet{DHY15}.  For the remainder of the discussion,
without loss of generality, we assume that the error estimates are
monotonic.

The following theorem shows that errors in estimated order statistics
propagate in a well-behaved fashion in surrogate-ranking mechanisms.

\begin{theorem}
\label{thm:propagation}
For all $i$ and $j$, let $\suri^j$ be the expected $j$th order
statistic of agent $i$'s virtual value distribution, and let
$\esuri^j$ be an estimate of $\suri^j$ satisfying
$|\esuri^j-\suri^j|<\errori$, where $\errori$ is an agent-specific
error bound.   The difference between the expected virtual
surplus of the surrogate ranking mechanisms with the true expected
order statistics and estimated order statistics is at most
$2\sum_i\errori$.
%
\end{theorem}
\begin{proof}
Let $\allocs$ and $\eallocs$ denote the allocation rule of the
surrogate-ranking mechanism as a function of agents' ranks $\ranks$
among their run-time samples with optimal surrogate values $\Surs$ and
estimated and ironed surrogate values $\eSurs$, respectively.  The
theorem follows from:
\begin{align*}
\expect[\ranks]{\sum\nolimits_i\texpect[\quanti]{\suri(\quanti)\,|\,\ranki}\,\ealloci(\ranks)}
& = \expect[\ranks]{\sum\nolimits_i\suri^{\ranki} \ealloci(\ranks)}\\
&\geq \expect[\ranks]{\sum\nolimits_i(\esuri^{\ranki}-\errori) \ealloci(\ranks)}\\
&\geq \expect[\ranks]{\sum\nolimits_i\esuri^{\ranki}\, \ealloci(\ranks)}-\sum\nolimits_i \errori\\
&\geq \expect[\ranks]{\sum\nolimits_i\esuri^{\ranki}\, \alloci(\ranks)}-\sum\nolimits_i \errori\\
&\geq \expect[\ranks]{\sum\nolimits_i( \suri^{\ranki} -\errori)\,\alloci(\ranks)}-\sum\nolimits_i \errori\\
&\geq \expect[\ranks]{\sum\nolimits_i\suri^{\ranki}\,\alloci(\ranks)}-2\sum\nolimits_i \errori.
\end{align*}
The second and fifth lines follow from the assumption that
$|\esuri^j-\suri^j|<\errori$, and the fourth line follows from the
fact that $\eallocs$ is the allocation rule that maximizes
$\texpect[\ranks]{\sum\nolimits_i\esuri^{\ranki}\,
  \ealloci(\ranks)}$. The last line is the expected virtual surplus of
the optimal surrogate-ranking mechanism $\allocs$, which implies the
result.
\end{proof}

\subsection{Sample Complexity of Non-Truthful Mechanisms}
\label{sec:instantiation}

We can now formalize the reduction from non-truthful sample complexity
to inference in rank-based position auctions. In
\Cref{sec:rank-based}, we observed that the data generated by
an agent in a SRM is distributed according to the unique BNE of an
i.i.d.\ position auction. In \Cref{sec:BAR}, we demonstrated the
existence of SRMs with near-optimal welfare and revenue, and in
\Cref{sec:optimal} and \Cref{sec:prop}, we showed that it was
possible to construct such a mechanism from noisy estimates of
expected order statistics. We have therefore reduced the problem designing a near-optimal non-truthful mechanism from data to the problem of estimating expected order statistics in an i.i.d.\ position auction. We now instantiate the reduction in two settings: bounded distributions with the winner-pays-bid or all-pay format, where we seek additive error bounds, and unbounded distributions with truthful payments, where we seek multiplicative error bounds.



We first consider bounded distributions with the winner-pays-bid or all-pay format. Existing literature shows how to one use bid data from SRMs to infer these parameters
efficiently. \citet{CHN16} study the problem of inferring order statistics from bid
data in all-pay and winner-pays-bid i.i.d.\ position auctions. They
show that for any non-trivial position weights, it is possible to
efficiently infer order statistics for both the values and marginal
revenues. We summarize their results below:

\begin{theorem}[\citealp{CHN16}]\label{thm:chnap}
Consider a $\num$-agent all-pay or winner-pays-bid i.i.d.\ position
auction with arbitrary position weights and values in $[0,1]$. There
exists an estimator $\hat V_k$ for the expected $k$th order statistic
of the value distribution $V_k$ such that with $N\geq
O(T^4(\log^2(1/\delta+T)+\log^2(1/\epsilon+
T))\epsilon^{-2}\delta^{-2})$ sampled bids from the unique BNE, $|\hat
V_k-V_k|\leq \epsilon$ with probability at least $1-\delta$.
\end{theorem}
\begin{theorem}[\citealp{CHN16}]\label{thm:chnfp}
Consider a $\num$-agent all-pay or winner-pays-bid i.i.d.\ position
auction with arbitrary position weights and values in $[0,1]$. There
exists an estimator $\hat \Psi_k$ for the expected $k$th order
statistic of the virtual value distribution $\Psi_k$ such that with $N\geq
O(T^4(\log^2(1/\delta)+\log^2(1/\epsilon))\epsilon^{-2}\delta^{-2})$
sampled bids from the unique BNE, $|\hat \Psi_k-\Psi_k|\leq \epsilon$
with probability at least $1-\delta$.
\end{theorem}

We note that the above results combine with the incentive analysis of \Cref{sec:rank-based}, approximation analysis of \Cref{sec:BAR}, and the analysis of error propagation in \Cref{sec:optimal} and \Cref{sec:prop} to 
imply a solution to the non-truthful sample complexity problem for
agents with values in $[0,1]$ and additive loss. We summarize below:

\begin{theorem}\label{thm:nontruthful}
For agents with bounded values in $[0,1]$, there are families of
winner-pays-bid and all-pay mechanisms that satisfy
\Cref{cond:inference}, \Cref{cond:incentives}, and
\Cref{cond:performance} with $p_{\text{run}}(n,\epsilon^{-1})= O(n^4\epsilon^{-3})$ and $p_{\text{design}}(n,\epsilon^{-1},\delta^{-1})=\tilde O(n^{28}\epsilon^{20}\delta^2)$ for additive loss and the welfare
objective. If the agents' distributions are regular and bounded, then
the same result also holds for the revenue objective.
\end{theorem}
\begin{proof}
Because values are bounded, a multiplicative error of $\epsilon$ corresponds to an additive error of at most $n\epsilon$. Hence, choosing $\num=n^4\epsilon^{-3}$ guarantees additive representation error $\epsilon$, by \Cref{THM:WELARE}. Next, note that to obtain an additive estimation error of $\epsilon$, \Cref{thm:propagation} implies that estimating each surrogate value to additive error at most $\epsilon/n$ suffices. Finally, to obtain these guarantees with probability at least $1-\delta$, the union bound implies that a failure probability of $n^{-1}\num^{-1}\delta=n^{-5}\epsilon^{3}\delta$ per surrogate value suffices. Theorems~\ref{thm:chnap} and~\ref{thm:chnfp} therefore imply that $\tilde O(n^4\epsilon^{-3})^4(\epsilon/n)^{-2}(n^{-5}\epsilon^{3}\delta)^{-2})=\tilde O(n^{28}\epsilon^{20}\delta^2)$ design-time samples suffice.
\end{proof}

We conclude by discussing our results' implications for the truthful sample complexity
literature. Our results yield polynomial sample
complexity for revenue maximization with unbounded regular
distributions and general feasibility settings.  This extends the
result of \citet{DHP16} by dropping the downward-closure requirement
on the feasibility constraint. Details can be found in \Cref{APP:SAMPLES}.



\begin{theorem}\label{thm:truthful}
For agents with regularly distributed (but potentially unbounded) values, there is a family of truthful mechanisms that satisfies
\Cref{cond:inference}, \Cref{cond:incentives}, and
\Cref{cond:performance} with $p_{\text{run}}(n,\epsilon^{-1})= O(n\epsilon^{-3})$ and $p_{\text{design}}(n,\epsilon^{-1},\delta^{-1})=\tilde O(n^{10}\epsilon^{-22})$ for multiplicative approximation and the revenue
objective.
\end{theorem}

We conclude by noting that the polynomials in the sample complexity guarantees of
\Cref{thm:nontruthful} and \Cref{thm:truthful} are clearly impractical. We leave optimizing the polynomials and deriving tight tradeoffs between run-time and design-time samples to future work.
	
	\bibliographystyle{apalike}
	\bibliography{bibs}

\appendix
\section{Undoing the Revelation Principle}

\label{app:unrevelation}

Good first-price and all-pay mechanisms for a given environment can be
found by undoing the revelation principle (ignoring computational
complexity).  This construction applies to any revelation mechanism
$\mathcal M$.  For concreteness, imagine a applying this approach to a
single-minded combinatorial auction problem where $\mathcal M$ is the
Vickrey-Clarke-Groves (VCG) mechanism.  We give the all-pay version of the
construction which is slightly simpler, but exhibits the same issues.

\begin{definition}
{\em The all-pay unrevelation mechanism} for a revelation
mechanism $\mathcal M$ is:
\begin{enumerate}
\item 
For each agent $i$ and value $v_i$, calculate $s_i(v_i)$ as the
expected payment in $\mathcal M$ when the agent's value is $v_i$
and other agents' values are drawn from the distribution.

\item 
For each agent $i$, given bid $b_i$ in the un-revelation mechanism,
calculate the agent's value as $v_i = s_i^{-1}(b_i)$.

\item Serve the agents who are served by $\mathcal M$ on values $\mathbf v = (v_1,\ldots,v_n)$; all agents pay their bids.
\end{enumerate} 
\end{definition}

The characterization of Bayes-Nash equilibrium
(Theorem~\ref{thm:myerson}) implies that $s_i$ is the strategy that
agents will employ in equilibrium of the constructed all-pay
mechanism.  Thus, the all-pay mechanism has the same equilibrium
outcome.

From this definition we can see why symmetric and ordinal environments
(i.e., IID position environments) are special.  For these environments
all agents will have the same strategy function, this strategy
function will order higher valued bidders higher (by monotonicity),
and the ordinal environment then implies that all that is needed to
select an outcome is the order of values not their cardinal values.
Thus, the mechanism simplifies to simply ordering the bids and the
strategy function does not need to be calculated.

Even absent computational issues in estimating the strategy functions
so as to implement this mechanism, it is clear that very detailed
distributional information is needed to run the unrevelation
mechanism.  Moreover, the resulting outcomes may be very sensitive to
small errors with the inversion of the strategy function.  This
unrevelation mechanism is not to be considered practical.

\section{Proof of \Cref{LEM:RANKVPRICE}}
\label{APP:RVP}

\begin{numberedlemma}{\ref{LEM:RANKVPRICE}}
The top $k$-of-$\num$ mechanism attains at least a $(1-O(k^{-1/2}))$-fraction of the welfare of the $k/\num$-price posting algorithm with $\num$ i.i.d. agents. If values are regularly distributed, then it attains at least a $(1-O(\min(k,T-k)^{-1/2}))$-fraction of the revenue of the $k/\num$-price posting mechanism.
\end{numberedlemma}

\begin{proof}

We argue separately for the objectives of welfare and virtual surplus, but in both cases, the proof strategy will be the same. We will first explicitly characterize the worst-case distribution for each objective, and then we will analyze the performance ratio of the top-$k$-of-$\num$ mechanism and $k/\num$-price posting mechanism using the correlation gap approach of \citet{Y11} or a similar analysis. In what follows, we will suppress subscripts denoting a particular agent when the agent's identity is irrelevant.
	
	Key to the analysis will be two formulae for the expected surplus of
	an mechanism, in terms of its interim allocation rule $x(\cdot)$ and
	the distribution's value function $\val(\cdot)$. We have that a
	mechanism's surplus is:
	\begin{equation}
	\label{eq:vsurplus}
	\mathbb E_{\quant\sim U[0,1]}[x(\quant)\val(\quant)]=\mathbb E_{\quant\sim U[0,1]}[-x'(\quant)V(\quant)],
	\end{equation}
	where $V(q)=\int_0^\quant \val(z)\,dz$, and the equality follows from
	integration by parts. This is the welfare analog of \Cref{lem:vvals}. We will make use of the latter result for revenue analysis. The only real
	difference between the two objectives is the fact that values are
	always positive, whereas virtual values may be negative.  This will change
	the nature of the approximation, as allocating the wrong agent becomes
	actively harmful to the performance of the algorithm in the case of revenue.\\
	
\textbf{Welfare.} We begin by normalizing the per-agent surplus of the price-posting mechanism to 1. Note that for the $k/\num$-price posting algorithm, the allocation rule is 1 until quantile $k/\num$, and then drops to $0$. It follows from equation (\ref{eq:vsurplus}) that our normalization is equivalent to the assumption that $V(k/\num)=1$. Next, we note that because $\val(\cdot)$ is positive and decreasing, $V(\cdot)$ is increasing and concave, with $V(0)=0$. Let $x(\cdot)$ be the allocation rule of the top-$k$-of-$\num$ algorithm. Given our normalization, the problem of finding the worst-case distribution then becomes:
	\begin{align*}
	\min_{V(\cdot)} &\,\,\,\,\,\mathbb E_{\quant\sim U[0,1]}[-x'(\quant)V(\quant)]\\
	\text{subject to}& \,\,\,\,\,V(0)=0\\
	& \,\,\,\,\,V(k/\num)=1\\
	& \,\,\,\,\,V(\cdot)\text{ concave}\\
	& \,\,\,\,\,V(\cdot)\text{ increasing}
	\end{align*}
	
	This program can be solved by inspection by noticing that there is pointwise minimal function satisfying the constraints of the program: namely, the optimal $V(\quant)$ is linear with slope $v(\quant)=\num/k$ for $\quant\leq k/\num$, and constant at 1 for $\quant\geq k/\num$. This corresponds to the distribution with $k/\num$ mass on the value $\num/k$, and the rest on 0. The result then immediately follows from the correlation gap for $k$-uniform matroids \citep{Y11}.\\
	
	
\textbf{Virtual Surplus.} We now adapt the above proof to virtual surplus. We will again characterize the worst-case distribution, but this time, we cannot simply apply the correlation gap as we did for welfare. As before, we normalize the per-agent virtual surplus from price-posting to 1. This corresponds with setting $R(q)=1$. Subject to normalization, we use properties of revenue curves to derive the worst-case distribution for virtual surplus. We assume values are regularly distributed, which implies that $R(q)$ is concave. Moreover, since $R(q)=v(q)q$, we have that $R(0)=R(1)=0$. These properties yield the following program for the worst-case distribution:
	\begin{align*}
	\min_{R(\cdot)} &\,\,\,\,\,\mathbb E_{q\sim U[0,1]}[-x'(q)R(q)]\\
	\text{subject to}& \,\,\,\,\,R(0)=R(1)=0\\
	& \,\,\,\,\,R(k/\num)=1\\
	& \,\,\,\,\,R(\cdot)\text{ concave}
	\end{align*}
	
	Again, this may be solved by inspection. The worst-case $R(\cdot)$ is triangular, with its apex at $(k/n,1)$. That is, on $[0,k/\num]$, $R(q)$ has slope $\num/k$, and on $[k/\num,1]$, it has slope $-\num/(\num-k)$. In other words, the worst-case distribution for virtual surplus has virtual value $\num/k$ with probability $k/\num$, and virtual value $-\num/(\num-k)$ otherwise.
	
	Let $X^+$ (resp. $X^-$) denote the number of agents with positive (resp. negative) virtual value. Further let $Y^+$ (resp. $Y^-$) denote the number of agents with positive (resp. negative) virtual value allocated by the top-$k$-of-$T$ mechanism, and $Z^+=X^+-Y^+$ (resp. $Z^-=X^--Y^-$) denote the number of agents with positive (resp. negative) virtual value who go unallocated. Since the expected virtual value of any individual agent is $0$, we may write:
\begin{equation*}
    \ex{\tfrac{T}{k}X^+-\tfrac{T}{T-k}X^-}=0.
\end{equation*}
This yields two equivalent expressions for the virtual surplus of the top-$k$-of-$T$ mechanism:
\begin{align}
    \ex{\text{Rev($k$,$T$)}}&=\tfrac{T}{k}\ex{Y^+}-\tfrac{T}{T-k}\ex{Y^-}\label{eq:pos}\\
    &=\tfrac{T}{T-k}\ex{Z^-}-\tfrac{T}{k}\ex{Z^+}.\label{eq:neg}
\end{align}
We can break our analysis into two cases, based on whether $k\leq T-k$ or $k>T-k$. In the former case, we analyze \cref{eq:pos}, and in the latter, we analyze \cref{eq:neg}.
In either case, we can lower bound the first term using the correlation gap as we did for welfare. For the second term, we bound the loss by a new analysis. We argue the $k\leq T-k$ case below. The other case follows by a symmetric argument.

Assume $k\leq T-k$. We must first lower bound $\tfrac{T}{k}\ex{Y^+}$. Noting that the setup is identical to the welfare analysis implies that we may again apply the correlation gap result of \citet{Y11} to get a lower bound of $(1-O(k^{-1/2}))\num$. We next upper bound $\tfrac{T}{T-k}\ex{Y^-}$. The expected value of $Y^-$ can be written as:
\begin{equation}\label{eq:loss}
    \frac{T}{T-k}\sum_{j=0}^{k-1} \binom{T}{j}\bigg(\frac{k}{T}\bigg)^j\bigg(\frac{T-k}{T}\bigg)^{T-j}(k-j).
\end{equation}
For any fixed $k$ and $j$, the summand in \cref{eq:loss} is an increasing function of $\num$. We may therefore upper bound it by its limit as $T\rightarrow \infty$ to obtain:
\begin{align*}
    \frac{T}{T-k}\ex{Y^-}&\leq \frac{T}{T-k}\sum_{j=0}^{k-1}\frac{k^j(k-j)}{e^kj!}\\
    &=\frac{T}{T-k}\bigg[\sum_{j=0}^{k-1}\frac{k^{j+1}}{e^kj!}-\sum_{j=1}^{k-1}\frac{k^{j}}{e^k(j-1)!}\bigg]\\
    &=\frac{T}{T-k}\cdot \frac{k^k}{e^k(k-1)!}\\
    &\leq T\Big(\frac{k^k}{e^kk!}\Big)=T\cdot O(k^{-1/2})
\end{align*}
The last inequality follows from the assumption that $k\leq T-k$, and the last equality from Stirling's approximation.
\end{proof}
\section{Proof of \Cref{LEM:TOPPROMOTION}}
\label{APP:BUFFER}

\begin{numberedlemma}{\ref{LEM:TOPPROMOTION}}
For each agent $i$, let $\phi_i:[0,1]\rightarrow \mathbb R$ be an arbitrary nonincreasing virtual value function satisfying $\int_0^1 \phi_i(q)\,dq\geq 0$. The $q$-buffering rule for $\allocs$ attains at least a $(1-\frac{q}{1-q})(1-q)(1-2(n-1)q)$-fraction
	of the virtual surplus of $\allocs$.
\end{numberedlemma}

The proof of Lemma~\ref{lem:toppromotion} will proceed in two main steps. First, we will show that applying the quantile remapping procedure in Definition~\ref{def:hedge} to a single agent $i$ (leaving other agents' quantiles untouched) cannot reduce the expected virtual surplus from that agent by too much. This will follow from \Cref{lem:apx}, which relates the virtual surpluses of allocation rules with inverses that are multiplicatively close. Second, we will show that subsequently applying the quantile remapping procedure to the populations other than $i$ also does not significantly reduce the expected virtual surplus from $i$. This will follow from the fact that the distribution of quantiles input to the base allocation algorithm is identical, conditioned on no agents having extreme quantiles. 

We begin with the single-agent analysis. Note that for a single agent, the $q$-buffering procedure can be thought of as two composed steps. First is a \emph{top promotion} procedure, which remaps sufficiently low quantiles to $0$ while remapping the remaining quantiles to induce a uniform distribution over $[0,1]$. Top promotion is then composed with \emph{bottom demotion}, which performs analogous transformation, mapping high quantiles to $1$ and mapping the rest of the interval to $[0,1]$. We formalize this as follows:

\begin{definition}
	Given a monotone single-agent allocation rule $x$ and quantile $\underline q$, the \emph{top promotion algorithm} for $x$ and $\underline q$  runs $x$ on the agent with quantiles transformed as follows:
	\begin{itemize}
		\item For any $q\in[0,\underline q]$, return $0$.
		\item For any $q\in[\underline q,1]$, return $(q-\underline q)/(1-\underline q)$.
	\end{itemize}
\end{definition}

\begin{definition}
	Given a monotone single-agent allocation rule $x$ and quantile $\overline q$, the \emph{bottom demotion algorithm} for $x$ and $\overline q$ runs $x$ on the agent with quantiles transformed as follows:
	\begin{itemize}
		\item For any $q\in [0,\overline q]$, return $q/\overline q$.
		\item For any $q\in [\overline q, 1]$, return $1$.
	\end{itemize}
\end{definition}

The interim allocation rule faced by an agent $i$ after applying the extremal buffering procedure to just $i$ is the composition of the bottom demotion algorithm for quantile $1-q$ composed with the top promotion algorithm for original allocation rule $x_i$ and quantile $q/1-q$.
Consequently, we may analyze the loss from applying these two transformations separately and multiply the losses.

We first analyze bottom demotion. While bottom demotion does not
produce an allocation rule which is multiplicatively close to the
original rule, it does produce one which is close in the sense that
its inverse is close to the inverse of the original rule. We may therefore apply  \Cref{lem:apx}
from \Cref{sec:BAR}, which we restate here for convenience.  Recall, a virtual value
function is $\virt(\cdot)$ and has cumulative virtual curve
$\cumvirt(\quant) = \int_0^\quant \virt(r)\,dr$.

\begin{numberedlemma}{\ref{lem:apx}}
For virtual value function $\virt(\cdot)$ and cumulative virtual value
$\cumvirt(q)=\int_0^{\quant} \virt(r)\,dr$ satisfying $\cumvirt(\alpha
\,\quant)\geq \alpha\, \cumvirt(\quant)$ for all quantiles $\quant$
and $\alpha \in [0,1]$, and any two allocation rules $\alloc_1$ and
$\alloc_2$ that satisfy $\alloc_1^{-1}(z)\geq \alloc_2^{-1}(z)\geq
\tfrac{1}{\alpha}\alloc_1^{-1}(z)$, the virtual surpluses satisfy
	\begin{align*}
	\expect[\quant]{\virt(\quant)\,\alloc_2(\quant)} + \cumvirt(0)\,\alloc_2(0)
   &\geq \tfrac{1}{\alpha} [\expect[\quant]{\virt(\quant)\,\alloc_1(\quant)} + \cumvirt(0)\,\alloc_1(0). 
	\end{align*} 
\end{numberedlemma}

\begin{lemma}\label{lem:bottomdemote}
	Let $\phi:[0,1]\rightarrow\mathbb R$ be an arbitrary nonincreasing virtual value function. Given a monotone single-agent allocation rule x and quantile $\overline q$, the bottom demotion algorithm for $x$ and $\overline q$ obtains at least a $\overline q$-fraction of the expected virtual surplus of $x$.
\end{lemma}
\begin{proof}
	The lemma will follow from a straightforward application of Lemma~\ref{lem:apx}. For a quantile $q$ receiving allocation $x(q)$ from the base algorithm, the quantile receiving this probability of allocation under the bottom demotion algorithm will be $\overline q q$. Hence, $x^{-1}(z)\geq \hat x^{-1}(z)=\overline q x^{-1}(z)$. Since $\phi$ is nonincreasing in $q$, we have that $R(q)=\int_0^q \phi(r)\,dr$ satisfies $R(\alpha q)\geq \alpha R(q)$ for all $\alpha\in [0,1]$. Lemma~\ref{lem:apx} therefore implies the desired result.
\end{proof}

We have shown that bottom demotion results in an allocation rule which has an inverse close to that of the original rule on which it is based. To derive an approximation result for the top promotion procedure requires 
 a more nuanced version of the same approach, based on two observations.
First, the ``unallocation rules'', i.e., $y(q) = 1-x(1-q)$ for allocation
rule $x(q)$, satisfy the inverse-approximation condition of the lemma.
Second, the virtual surplus of the unallocation rule is given by the
expected virtual value plus the negative virtual surplus of the
unallocation rule.  Specifically ${\mathbb E}_q[\phi(q)\,x(q)] =
{\mathbb E}_q[\phi(q)] + {\mathbb E}_q[(-\phi(1-q))\,y(q)]$.
While virtual values for revenue always satisfy the property that rays
from the origin cross the cumulative virtual value curve from below,
this property does not generally hold for the negative virtual values
$-\phi(1-q)$.  Regularity, i.e., monotonicity of the original virtual
value function, however, implies the property for negative virtual
values.  These observations are formally summarized in the subsequent
lemma:

\begin{lemma}\label{lem:toppromote}
	Let $\phi:[0,1]\rightarrow \mathbb R$ be a nonincreasing virtual value function with $\int_0^1 \phi(q)\,dq\geq 0$. Given a monotone single-agent allocation rule $x$ and quantile $\underline q$, the top promotion algorithm for $x$ and $\underline q$ obtains at least a $(1-\underline q)$-fraction of the expected virtual surplus of $x$.
\end{lemma}
\begin{proof}
Note that the expected virtual surplus from any allocation rule $x$ is can be written as $\int_0^1 \phi(q)x(q)\, d q=\int_0^1 \phi(q)\,dq-\int_0^1\phi(q)(1-x(q))\,dq$. Specifically, let $\hat x$ be the allocation rule of the top promotion algorithm, and $x$ the allocation rule of the original algorithm. Moreover, define $\hat y(q)=1-\hat x(1-q)$ and $y(q)=1-x(1-q)$ to be the corresponding ``unallocation rules.'' We will show that 
	\begin{equation}
	\label{eq:dealloc}
	\int_0^1-\phi(1-q)\hat y(q)\,dq \geq (1-\underline q) \int_0^1-\phi(1-q)y(q)\,dq.
	\end{equation}
	Since, $\int_0^1 \phi(q)\,dq\geq 0$, this will prove that $\int_0^1 \phi(q)x(q)\, d q\geq (1-\underline q)\int_0^1 \phi(q)\hat x(q)\, d q$.
	
	To prove (\ref{eq:dealloc}), note that the definition of the top promotion algorithm can be manipulated to obtain $\hat x^{-1}(z)=x^{-1}(z)(1-\underline q)+\underline q$. Moreover, by the definition of $y$ and $\hat y$, we have $y^{-1}=1-x^{-1}(1-z)$ and $\hat y^{-1}=1-\hat x^{-1}(1-z)$. Combining these three equations yields that $y^{-1}(z)\geq \hat y^{-1}(z)= (1-\underline q) y^{-1}(z)$ for all $z\in[0,1]$. Moreover, note that $-\phi(1-q)$ is decreasing in $q$. This implies that $R(q)/q\geq -\phi(1-q)$, where $R(q)=\int_0^q -\phi(1-q)\,dq$. We may therefore apply Lemma~\ref{lem:apx}, which yields (\ref{eq:dealloc}).
\end{proof}

Combining Lemmas~\ref{lem:bottomdemote} and~\ref{lem:toppromote} yields the following:
\begin{lemma}\label{lem:singleagent}
	Let $\phi:[0,1]\rightarrow \mathbb R$ be an arbitrary nonincreasing virtual value function satisfying $\int_0^1 \phi(q)\,dq\geq 0$, and consider an arbitrary agent $i$. The $q$-buffering algorithm for ${\mathbf x}$ and $\hat Q$, when applied only to agent $i$, attains at least a $(1-\frac{q}{1-q})(1-q)$-fraction of the expected virtual surplus for $i$.
\end{lemma}

Having derived a single-agent guarantee, we now show that applying the $q$-buffering procedure to all agents at once, rather than just to one agent, yields only a small additional loss. Intuitively, for each agent, the mechanism only appears different when another agent has an extreme quantile which is promoted or demoted by the buffering procedure. The probability of such an event can be controlled using the union bound. Formally, we have:

\begin{proof}[Proof of Lemma~\ref{lem:toppromotion}]
	
	Lemma~\ref{lem:singleagent} states that the virtual surplus lost from applying the $q$-buffering procedure to a single agent is small. We now argue that applying the algorithm to all agents at once does not incur much additional loss. We argue from the perspective of some agent $i$. 
	
	The key observation in our analysis is
	that the distribution of the quantiles of other agents input to $\allocs$ is nearly unchanged by the extremal buffering procedure.  In particular, note
	that the probability that one or more agents other
	than $i$ with quantiles set to $0$ or $1$ by the extremal buffering procedure is at
	most $(n-1)(1-2q)$, by the union bound.  Conditioned on there being no
	such agents, the distribution of quantiles input to $\allocs$ remains uniform. It follows that the virtual surplus from agent $i$
	conditioned on this event is identical to the virtual surplus from
	the extremal buffering procedure applied only to $i$.
	
	In the event that there are one
	or more agents from populations other than $i$ who have top quantiles (which are promoted) or bottom quantiles (which are demoted), we note that the conditional virtual surplus from population $i$ is nonnegative. To see this, let $\tilde x_i$ be the $q$-buffered interim allocation rule for agent $i$, conditioned on the event $\mathcal E$ that at least one agent $j$ other than $i$ has a quantile in $[0,q]\cup[1-q,1]$. Since $\allocs$ is a monotone function of its inputs, it must be that $\tilde x_i$ is  nondecreasing. The expected virtual surplus from agent $i$ conditioned on $\mathcal E$ is $\int_0^1 \phi_i(q)\tilde x_i(q)\,dq$. By assumption, $\int_0^1 \phi_i(q)\,dq\geq 0$, so it must also be the case that $\int_0^1 \phi_i(q)\tilde x_i(q)\,dq\geq 0$.
	
	To conclude the proof, let $\hat x_i$ denote the interim allocation rule after extremal buffering, conditioned on the event $\overline{\mathcal E}$. The total virtual surplus from agent $i$ is:
	\begin{equation}
	\notag
	\text{Pr}(\mathcal E)\int_0^1 \phi_i(q)\tilde x_i(q)\,dq+\text{Pr}(\overline{\mathcal E})\int_0^1 \phi_i(q)\hat  x_i(q)\,dq
	\end{equation}
	By the union bound, $\text{Pr}(\overline{\mathcal E})=1-\text{Pr}(\mathcal E)\geq 1-(n-1)(1-2q).$ By Lemma~\ref{lem:singleagent} and the fact that, conditioned on $\overline{\mathcal E}$, the distribution of quantiles $i$ perceives from other agents is uniform implies that 
	\begin{eqnarray}
	\notag \int_0^1 \phi_i(q)\hat x_i(q)\,dq&\geq (1-\frac{q}{1-q})(1-q)\int_0^1 \phi_i(q)x_i(q)\,dq\\
	\notag &\geq (1-\frac{q}{1-q})(1-q)\int_0^1 \phi_i(q)x_i(q)\,dq.
	\end{eqnarray}
	Since we have shown that $\int_0^1 \phi_i(q)\hat  x_i(q)\,dq\geq 0$, we can combine the above to obtain:
	\begin{eqnarray}
	\notag
	\text{Pr}(\mathcal E)\int_0^1 \phi_i(q)\tilde x_i(q)\,dq+\text{Pr}(\overline{\mathcal E})\int_0^1 \phi_i(q)\hat  x_i(q)\,dq\\
	\notag \geq \big(1-\frac{q}{1-q}\big)(1-q)(1-2(n-1)q)\int_0^1\phi_i(q)x_i(q)\,dq.
	\end{eqnarray}
	Summing over agents proves the lemma.
\end{proof}

Since we have reasoned about abstract virtual surplus, which could be value or Myerson virtual value, we obtain revenue and welfare approximation results for extremal buffering.
\section{Proof of \Cref{lem:resampling}}
\label{APP:RESAMPLE}

\begin{numberedlemma}{\ref{lem:resampling}}
For each agent $i$, let $\phi_i(q_i)$ be a nonincreasing virtual value function, and let $\mathbf x$ be an a monotone allocation algorithm. Further, for some $k\leq T/2$, let $\mathbf{\hat x}$ denote the $k/T$-buffering algorithm for $\mathbf x$. Then the resampling algorithm for $\mathbf{\hat x}$ obtains at least a $\frac{k}{k+1}(1-\frac{q}{1-q})(1-q)(1-2(n-1)q)$ the expected virtual surplus under $\mathbf x$, with $q=k/T$.
\end{numberedlemma}

To prove Lemma~\ref{lem:resampling}, we first rephrase a useful lemma from \citet{RS16}, which characterizes the relationship between the virtual surplus of resampling rules and their base allocation rules. Intuitively, it states that resampling can be treated as a distributional transformation - rather than replacing the allocation on each interval with its average, you could think about replacing the virtual values on each interval with their average. Formally:

\begin{lemma}[\citealp{RS16}]\label{lem:switching}
	For every agent $i$, let $\phi_i$ be a virtual value function for each agent, with cumulative virtual surplus $R_i(q)=\int_0^q\phi_i(r)\,dr$. For any allocation algorithm $\allocs$, let $\overline{\allocs}$ denote the resampling algorithm based on $\allocs$. Then for each agent $i$ we have:
	\begin{equation}
	\mathbb E[-\overline x_i'(q_i)R_i(q_i)]=\mathbb E[-\overline x_i'(q_i)\overline R_i(q_i)]=\mathbb E[- x_i'(q_i)\overline R_i(q_i)],
	\notag
	\end{equation}
where $\overline R_i(q_i)$ is a piecewise linear approximation to $R_i$ given by:
\begin{equation}
\notag
\overline R_i(q)=T(q-\tfrac{j-1}{T})R_i(\tfrac{j}{T})+T(\tfrac{j}{T}-q)R_i(\tfrac{j-1}{T})\text{ for }q\in[\tfrac{j-1}{T},\tfrac{j}{T}].
\end{equation}
\end{lemma}

\begin{proof}[Proof of Lemma~\ref{lem:resampling}]
	Let $\hat{\allocs}$ denote the  $k/T$-buffering rule, and $\overline{\allocs}$ denote the resampling rule based on $\hat{\allocs}$. Define
\begin{equation}
\notag
\hat R_i(q_i)=\begin{cases}
R_i(q_i)&\text{ for }q_i\in[\tfrac{k}{T},\tfrac{T-k}{T}]\\
\tfrac{Tq_i}{k}R_i(\tfrac{k}{T})&\text{ for }q_i\in[0,\tfrac{k}{T}]\\
R_i(\tfrac{T-k}{T})+\tfrac{T}{k}(q-\tfrac{T-k}{T})(R_i(1)-R_i(\tfrac{T-k}{T}))&\text{ for }q_i\in[\tfrac{T-k}{T},1].
\end{cases}
\end{equation}
That is, $\hat R_i$ is equal to $R_i$ except on $[0,\tfrac{k}{T}]$ and $[\tfrac{T-k}{T},1]$, where it is a linear interpolation between the values of the $R_i$ at the endpoints of those intervals.
	
	We will argue the following sequence of inequalities for each agent:
	\begin{eqnarray}
	\notag\mathbb E[-\overline x_i'(q_i)R_i(q_i)]&=\mathbb E[-\hat x_i'(q_i)\overline R_i(q_i)]\\
	\notag&\geq \frac{k}{k+1}\mathbb E[-\hat x_i'(q_i)\hat R_i(q_i)]\\
	\notag&=\frac{k}{k+1}\mathbb E[-\hat x_i'(q_i)R_i(q_i)].
	\end{eqnarray}
	The first equality follows from Lemma~\ref{lem:switching}.
We will argue the second and third lines shortly.  Since the first and last expressions in the above chain are the respective virtual surpluses of the resampling and buffering rules, respectively, the result will follow.
	
	To see that $\overline R_i(q_i)\geq \tfrac{k}{k+1}\hat R_i(q_i)$, note that $\overline R_i(q_i)=\hat R_i(q_i)$ for all quantiles in $[0,\hat q_i^k]\cup[\hat q_i^{T-k},1]$. Otherwise, consider $q_i\in[\tfrac{j}{T},\tfrac{j+1}{T}]$ for $j\in\{k,\ldots,T-k-1\}$. Assume without loss of generality that $R_i(\tfrac{j}{T})\leq R_i(\tfrac{j+1}{T})$; a symmetric argument will apply to the case where $R_i(\tfrac{j}{T})\geq R_i(\tfrac{j+1}{T})$. The concavity of $R_i$ implies that for all $q\in[\tfrac{j}{T},\tfrac{j+1}{T}]$, $R_i(q)\leq \tfrac{Tq}{j}R_i(\tfrac{j}{T})$. Moreover, note that $\overline R_i(q)\geq R_i(\tfrac{j}{T})$ for all $q\in[\tfrac{j}{T},\tfrac{j+1}{T}]$. Since $\tfrac{Tq}{j}\leq \tfrac{k+1}{k}$, it follows that $\overline R_i(q_i)\geq \tfrac{k}{k+1} R_i(q_i)=\tfrac{k}{k+1}\hat R_i(q_i)$ for all $q_i\in[\tfrac{k}{T},\tfrac{T-k}{T}]$, and $\overline R_i(q_i)=\hat R_i(q_i)$ elsewhere.
    
    Finally, we argue that $\mathbb E[-\hat x_i'(q_i)\hat R_i(q_i)]=E[-\hat x_i'(q_i)R_i(q_i)]$. Note that since $\hat R_i(q_i)=R_i(q_i)$ for $q_i\in [\tfrac{k}{T},\tfrac{T-k}{T}]$, it follows that 
\begin{equation*}
    \int_{\tfrac{k}{T}}^{\tfrac{T-k}{T}} -\hat x_i(q_i)R_i(q_i)\, dq_i=\int_{\tfrac{k}{T}}^{\tfrac{T-k}{T}} -\hat x_i(q_i)\hat R_i(q_i)\, dq_i.
\end{equation*}
To prove the claim, notice that $\hat x_i'(q_i)=0$ on $[0,\tfrac{k}{T}]$ and $[\tfrac{T-k}{T},1]$. Hence,
\begin{equation*}
    \int_{0}^{\tfrac{k}{T}} -\hat x_i'(q_i)R_i(q_i)\, dq_i=\int_{0}^{\tfrac{k}{T}} -\hat x_i'(q_i)\hat R_i(q_i)\, dq_i.
\end{equation*}
and
\begin{equation*}
    \int_{\tfrac{T-k}{T}}^{1} -\hat x_i'(q_i)R_i(q_i)\, dq_i=\int_{\tfrac{T-k}{T}}^{1} -\hat x_i'(q_i)\hat R_i(q_i)\, dq_i.
\end{equation*}
This proves the lemma.
\end{proof}
\section{Proof of \Cref{THM:OPTIMALSRM}}
\label{app:optimal-srm}

\begin{numberedtheorem}{\ref{thm:optimalSRM}}
The welfare-optimal surrogate-ranking mechanism uses surrogate values given by
$\suri^j=\expect[\quanti]{\vali(\quanti) \given \ranki=j}$. For
regular distributions, the revenue-optimal surrogate-ranking mechanism
uses surrogate values $\suri^j=\expect[\quanti]{\margi(\quanti) \given
  \ranki=j}$ where $\margi(\quanti) = \vali(\quanti) +
\quanti\,\vali'(\quanti)$ is the marginal revenue of agent $i$ at
quantile $\quanti$.
\end{numberedtheorem}

\begin{proof}
We define the \emph{rank-based allocation problem} as follows: the
designer must choose an allocation rule $\ralloc$ which takes as input
the rank $\ranks = (\ranki[1],\ldots,\ranki[n])$ of each agent among
$T-1$ runtime samples for their distributions and outputs a (possibly
randomized) feasible allocation $\rallocs(\ranks)$. As a constraint,
$\rallocs$ must be monotone in the ranks of each agent. The objective
is to maximize $\expect{\sum_i\virti(\quanti)\ralloci(\ranks)}$ for
some given virtual value function $\virti(\cdot)$, where the
expectation is over agents' quantiles being uniformly distributed and
over the runtime samples used to compute $\ranks$. For example,
$\virti(\quanti)=\vali(\quanti)$ corresponds to welfare maximization
and $\virti(\quanti)=\margi(\quanti)$ corresponds to revenue
maximization.

The rank-based allocation problem can be solved by inspection. Fixing
an allocation rule, the objective can be rewritten as
$\sum_i\expect{\virti(\quanti)\given \ranki}\,\ralloci(\ranks)$ by
linearity of expectation. From this expression, it becomes clear that
the optimal solution chooses the allocation which maximizes the
quantity $\sum_i \expect{\virti(\quanti) \given
  \ranki}\,\ralloci(\ranks)$. Note that if $\virti(\cdot)$ is
monotone, then this allocation rule will be monotone as well, and
therefore feasible.\footnote{If it is not monotone, then then the
  resulting surrogate values may not be monotone, if the surrogate
  values by this approach are not monotone, they can be ironed using
  the standard procedure.}  Setting these expected order statistics as
surrogate values, the surrogate-ranking mechanism (\Cref{d:SRA}) optimizes this quantity.
\end{proof}

\section{Sample Complexity in General Feasibility Environments}
\label{APP:SAMPLES}

In this appendix, we show how to use polynomially many truthfully sampled values to estimate the revenue
of the $k$-unit, $T$-buyer auction for all $k$ from $1$ to $T-1$
simultaneously. We assume the value distribution is regular, but may have unbounded
support.

\begin{theorem}\label{thm:truthfulestimation}
Given a regular value distribution, let $R^*=\max_\quant \rev(\quant)$ be the {\em monopoly revenue}. For
any $\epsilon,\delta\in(0,1)$, $O(T^6\epsilon^{-4}\log(T/\delta))$ samples
suffice to estimate the expected revenue of a $k$-unit, $T$-bidder
auction up to additive error $\epsilon R^*$ for all
$k\in\{1,\ldots,T\}$ simultaneously with probability at least
$1-\delta$.
\end{theorem}

We may combine \Cref{thm:truthfulestimation} our bounds on the propagation of error and the representation error of surrogate ranking mechanisms (\Cref{thm:propagation} and \Cref{THM:WELARE}, respectively) to obtain the desired sample complexity result for truthful mechanisms, restated below. 

\begin{numberedtheorem}{\ref{thm:truthful}}
For agents with regularly distributed (but potentially unbounded) values, there is a family of truthful mechanisms that satisfies
\Cref{cond:inference}, \Cref{cond:incentives}, and
\Cref{cond:performance} with $p_{\text{run}}(n,\epsilon^{-1},\delta^{-1})=\tilde O(n\epsilon^{-3})$ and $p_{\text{design}}(n,\epsilon^{-1},\delta^{-1})=\tilde O(n^{10}\epsilon^{-22})$ for multiplicative approximation and the revenue
objective.
\end{numberedtheorem}
\begin{proof}
First, note that $\num=n\epsilon^{-3}$ surrogate values per agent suffice to obtain multiplicative representation error at most $\epsilon$, by \Cref{THM:WELARE}. Now let $R_i^*$ denote the monopoly revenue for agent $i$'s distribution, and let $\text{OPT}$ denote the optimal revenue. Estimating each surrogate value to an additive $\epsilon n^{-1}R_i$ suffices to obtain multiplicative estimation error of $O(\epsilon)$, as $\sum_i R_i^*\leq n \text{OPT}$. Finally, failure probability of $\delta/n$ suffices for each agent to obtain these surrogate value estimates with probability at least $1-\delta$, by the union bound. Hence, $\tilde O((n\epsilon^{-3})^6(\epsilon/n)^{-4})=\tilde O(n^{10}\epsilon^{-22})$ samples suffice to obtain multiplicative loss of $\epsilon$ with probability at least $1-\delta$.
\end{proof}

We now outline the high-level strategy for proving
\Cref{thm:truthfulestimation}. First, for any $j\in\{0,\ldots, T\}$,
let $P_j$ denote the expected revenue of a $j$-unit auction with $T$
agents, and let $\psi^k$ denote the expected $k$th order statistic of
the virtual value distribution. Then we may write:
$\psi^k=P_k-P_{k-1}$. It follows that to estimate $\psi^k$ with
additive error $\epsilon R^*$, it suffices to estimate $P_k$ and
$P_{k-1}$ with error $\epsilon R^*/2$. 

To estimate the $k$-unit revenue $P_k$, we will estimate the revenue contribution from a single agent, $\mathbb E_\quant [\price(\quant)]$. Note that from an agent's perspective, playing in a $k$-unit auction is equivalent to facing a posted price distributed according to $\val(\quant_{k:T-1})$, where $\quant_{k:T-1}$ denotes the $k$th lowest order statistic of $T-1$ $U[0,1]$ random variables. A basic property of the order statistics of uniform variables is that $\quant_{k:T-1}$ is distributed according to $\text{Beta}(k,T-k)$. Let $f_{k:T-1}$ denote the density of $\quant_{k:T-1}$, and let $\rev=\val(\quant)\quant$ denote the price-posting revenue curve. We have:

\begin{lemma}\label{lem:multiunit}
For any $k\in\{1,\ldots,T-1\}$, $P_k=T\int_0^1 f_{k:T-1}(\quant)\rev(\quant)\,d\quant$.
\end{lemma}

In what follows, we will show how to estimate $\rev(\quant)$ for all $\quant\in[1/T^2,1-1/T^2]$. We will further show that the loss from misestimating $\rev(\quant)$ on $[0,1/T^2]\cup[1-1/T^2,1]$ is minimal. This will immediately imply Theorem~\ref{thm:truthfulestimation}.

\subsection{Estimation on a Grid}
To create a skeleton for our estimated revenue curve, we first estimate $\rev(\quant)$ for $\quant\in \{1/K,\ldots,1-1/K\}$, for some large $K$ to be determined later. The concavity of $\rev$ will imply that the rest of the revenue can be estimated with low error via interpolation.

\begin{lemma}\label{lem:grid}
Let $R^*=\max_\quant \rev(\quant)$. For any $1>\epsilon>0$ and $\delta$, $O(K^2\epsilon^{-2}\log(K/\delta))$ samples suffice to guarantee that $K\hat \val_j/j\in[R(j/K)-\epsilon R^*,R(j/K)+\epsilon R^*]$ for all $j\in\{1,\ldots,K-1\}$ simultaneously with probability at least $1-\delta$.
\end{lemma}

Consider drawing $N=Km-1$ samples, for some positive integer $m$. Note that the $jm$th smallest sample has mean $j/K$. Let $\hat\val_j$ denote the value of this sample. We will use $\hat\val_j$ as an estimator for $\val(j/K)$, and $K\hat \val_j/j$ as an estimator for $\rev(j/K)=T\val(j/K)/j$. The proof will proceed in two steps. First, we will use a Chernoff bound to show that with high probability, $\quant(\hat \val_j)$ is close to $j/K$. We will then use the concavity of $\rev$ to show that $\hat\val_j$ is close to $v(j/K)$.

\begin{lemma}\label{lem:singlequant}
For any $1>\epsilon>0$ and $\delta$, $Km=O(K\epsilon^{-2}\log(1/\delta))$ samples suffice to guarantee that $q_j\in[(1-\epsilon)j/K,(1+\epsilon)j/K]$ with probability at least $1-\delta$.
\end{lemma}
\begin{proof}
We now bound the probability of significantly misestimating $\quant(\hat\val_j)$. Let $\quant_j=\quant(\hat \val_j)$. Note that for any $\epsilon\in(0,1)$, the number of samples with quantile at most $(1+\epsilon)j/K$ is the sum of $N$ iid Bernoulli random variables with mean $(1+\epsilon)j/K$. Moreover, note that $q_j>(1+\epsilon)j/K$ only if at most $jm-1$ samples overall have quantile at most $(1+\epsilon) j/K$. Chernoff then gives us that
\begin{equation}
\label{eq:chernoffover}
\text{Pr}[q_j\geq (1+\epsilon)j/K]\leq e^{-\left(1-\frac{(jm-1)K}{(1+\epsilon)Nj}\right)^2\frac{(1+\epsilon)Nj}{2K}}=e^{-\Omega(\epsilon^2mj)}
\notag
\end{equation}

Similarly, the the number of samples with quantile at most $(1-\epsilon)j/K$ is the sum of $N$ iid Bernoullis with mean $(1-\epsilon)j/K$. We have $q_j<(1-\epsilon)j/K$ only if at least $jm$ samples have quantile at most $(1-\epsilon)j/K$. Chernoff then gives us:
\begin{equation}
\label{eq:chernoffunder}
\text{Pr}[q_j\leq (1-\epsilon)j/K]\leq e^{-\Omega(\epsilon^2mj)}
\notag
\end{equation}
It follows that as $m=O(\epsilon^{-2}j^{-1}\log(1/\delta))$ suffices for suffices for $q_j\in[(1-\epsilon)j/K,(1+\epsilon)j/K]$ with probability at least $1-\delta$. This bound is worst when $j=1$. Hence $mK=O(K \epsilon^{-2}\log(1/\delta))$ samples suffice overall.
\end{proof}

We next show that if $q_j$ is close to $j/K$, then $K\hat \val_j/j$ will be a close estimate of $R(j/K)$.

\begin{lemma}\label{lem:singlerev}
Assume $q_j\in[(1-\epsilon)j/K,(1+\epsilon)j/K]$. Then $K\hat \val_j/j\in[(1-\epsilon K)R(j/K),(1+\epsilon K)R(j/K)]$.
\end{lemma}
\begin{proof}
We will show that $\hat \val_j\in[(1-\epsilon K)v(j/K),(1+\epsilon K)v(j/K)]$. The result follows from multiplying by $K/j$. We first argue the case where $q_j<j/K$. That is, $\hat \val_j\geq \val(j/K)$. Concavity of the revenue curve implies that $R(q_j)\leq \tfrac{1-q_j}{1-j/K}\rev(j/K)$. In the event that $q_j\geq (1-\epsilon)j/K$, we have:
\begin{align*}
\hat \val_jq_j\leq \frac{1-q_j}{1-j/K}\frac{v(j/K)j}{K}\leq (1+\epsilon K)v(j/K).
\end{align*}
Dividing the inequality by $q_j$ and using the fact that $q_j\geq (1-\epsilon)j/K$ yields:
\begin{equation}
\hat \val_j\leq \frac{(1-(1-\epsilon)j/K)}{(1-j/K)(1-\epsilon)}v(j/K).
\notag
\end{equation}
A symmetric argument applies when $q_j>j/K$. By concavity, we have $\rev(q_j)\geq \tfrac{1-q_j}{1-j/K}\rev(j/K)$. In the event that $q_j\leq (1+\epsilon)j/K$, we have
\begin{equation}
v_j\geq \frac{(1-(1+\epsilon)j/K)}{(1-j/K)(1+\epsilon)}v(j/K)\geq (1-\epsilon K)\val(j/K).
\notag
\end{equation}
\end{proof}
Lemma~\ref{lem:grid} follows from the above results by applying the union bound.
\begin{proof}[Proof of Lemma~\ref{lem:grid}]
Lemmas~\ref{lem:singlequant} and~\ref{lem:singlerev} imply that $O(K^2\epsilon^{-2}\log(K/\epsilon))$ samples suffice to guarantee that $K\hat \val_j/j\in[(1-\epsilon)R(j/K),(1+\epsilon)R(j/K)]$ with probability at least $1-\delta/K$. Applying the union bound, it follows that $O(K^2\epsilon^{-2}\log(K/\epsilon))$ samples suffice to guarantee that this guarantee holds for all $j\in\{1,\ldots,K-1\}$ simultaneously with probability at least $1-\delta$. Noting that $R(j/K)\leq R^*$ implies the lemma.
\end{proof}

\subsection{Estimating The Interior Revenue Curve}

In the previous section, we showed how to estimate $\rev(j/K)$ up to an additive $\epsilon R^*$. Pick some $q\in[(j-1)/K,j/K]$ with $j\in\{2,\ldots, K-1\}$. We may linearly interpolate between our estimates of $\rev((j-1)/K)$ and $\rev(j/K)$ to estimate $\rev(q)$. More formally, let $\hat \rev_j$ denote the estimator for $\rev(j/K)$ analyzed in the previous section. We will estimate $\rev(q)$ as
\begin{equation*}
\hat R(q)=(q-\tfrac{j-1}{K})K\hat R_j+(\tfrac{j}{K}-q)K\hat R_{j-1}
\end{equation*}
We will use concavity of $\rev$ to bound the error from this estimate.

\begin{lemma}\label{lem:interior}
Assume that for all $j\in\{1,\ldots, K-1\}$, $K\hat \val_j/j\in[R(j/K)-\epsilon R^*,R(j/K)+\epsilon R^*]$ for all $j\in\{1,\ldots,K-1\}$. Then for all $q\in[1/K,1-1/K]$:
\begin{equation*}
\hat \rev(q)\in[\rev(q)-(\epsilon+K^{-1}) \rev^*,\rev(q)+\epsilon \rev^*]
\end{equation*}
\end{lemma}
\begin{proof}
We bound the overestimation and underestimation error in turn. First, not that since $\rev$ is concave, it must be that $\rev(q)$ lies above the line between $((j-1)/K,\hat R_{j-1}-\epsilon R^*)$ and $(j/K,\hat R_{j}-\epsilon R^*)$. It follows that $\hat \rev(q)$ can only overestimate by at most $\epsilon R^*$. Next, note that a lower bound on $\hat \rev(q)$ is 
\begin{equation*}
\hat \rev(q)\geq (q-\tfrac{j-1}{K})KR(\tfrac{j-1}{K})+(\tfrac{j}{K}-q)KR(\tfrac{j}{K})-\epsilon R^*
\end{equation*}
In other words, the worst-case underestimation is $\epsilon R^*$, plus the worst-case underestimation of the piecewise linear curve through the points $(j/K,\rev(j/K))$ for all $j\in\{0,K\}$. Using the facts that $\rev$ is concave, $\rev(0)=0$, and $\rev(1)=0$, this underestimation can be shown to be $O(\rev^*/K)$. The lemma follows.

\end{proof}

\subsection{Proof of Theorem~\ref{thm:truthfulestimation}}

We have shown how to estimate the revenue curve to low additive error for quantiles bounded away from $0$ and $1$. We now use this estimator to prove the main sample complexity result of this appendix: that $O(T^6\epsilon^{-4}\log(T/\epsilon))$ samples suffice to estimate estimate $\psi^k$ for all $k\in\{1,\ldots, T-1\}$ to within additive error $\epsilon R^*$ with probability at least $1-\delta$. To do so, we will consider estimating $P_k$ for $k\in\{1,\ldots, T-1\}$ as $T\int_{\epsilon/T^2}^{1-\epsilon/T^2} f_{k:T-1}(\quant)\hat \rev(\quant)\,d\quant$.

First we show that ignoring the revenue contribution from the intervals $[0,\epsilon/T^2]$ and $[1-\epsilon/T^2,1]$ cannot hurt our estimate by much. Let $F_{k:T-1}$ denote the CDF of the $k$th lowest order statistic of $T-1$ $U[0,1]$ random variables. Then using properties of the Beta distribution, we have that $F_{k:T-1}(\epsilon/T^2)\leq \epsilon/T$ and $1-F_{k:T-1}(1-\epsilon/T^2)\leq \epsilon/T$ for all $k\in\{1,\ldots,T-1\}$. Since for $\quant\in[0,\epsilon/T^2]\cup[1-\epsilon/T^2,1]$, $\rev(\quant)\leq R^*$, it follows from \Cref{lem:multiunit} that
\begin{equation*}
P_k-T\int_{\epsilon/T^2}^{1-\epsilon/T^2} f_{k:T-1}(\quant) \rev(\quant)\,d\quant\leq \epsilon R^*
\end{equation*}
Now assume that for all $\quant\in[\epsilon/T^2,1-\epsilon/T^2]$, $|\hat R(q)-R(q)|\leq \epsilon T^{-1} R^*$. If this is the case, then we have 
\begin{equation*}
\left |T\int_{\epsilon/T^2}^{1-\epsilon/T^2} f_{k:T-1}(\quant) \rev(\quant)\,d-T\int_{\epsilon/T^2}^{1-\epsilon/T^2} f_{k:T-1}(\quant) \hat \rev(\quant)\,d\quant\right |\leq \epsilon R^*.
\end{equation*}
Choosing $K=T^2/\epsilon$ in \Cref{lem:interior} therefore yields the result.

\end{document}